\documentclass[aps,twocolumn,superscriptaddress,floatfix,pra]{revtex4-1}

\usepackage[utf8]{inputenc}
\usepackage{amsmath,amsthm}
\usepackage{amssymb}
\usepackage[colorlinks=true,citecolor=blue,linkcolor=blue,urlcolor=blue]{hyperref}
\usepackage[capitalise]{cleveref}

\usepackage{times}
\usepackage[english]{babel}

\usepackage{array}

\makeatletter
\def\bbl@set@language#1{%
	\edef\languagename{%
		\ifnum\escapechar=\expandafter`\string#1\@empty
		\else\string#1\@empty\fi}%
	\@ifundefined{babel@language@alias@\languagename}{}{%
		\edef\languagename{\@nameuse{babel@language@alias@\languagename}}%
	}%
	\select@language{\languagename}%
	\expandafter\ifx\csname date\languagename\endcsname\relax\else
	\if@filesw
	\protected@write\@auxout{}{\string\select@language{\languagename}}%
	\bbl@for\bbl@tempa\BabelContentsFiles{%
		\addtocontents{\bbl@tempa}{\xstring\select@language{\languagename}}}%
	\bbl@usehooks{write}{}%
	\fi
	\fi}
\newcommand{\DeclareLanguageAlias}[2]{%
	\global\@namedef{babel@language@alias@#1}{#2}%
}
\makeatother

\DeclareLanguageAlias{en}{english}
\DeclareLanguageAlias{English}{english}
\DeclareLanguageAlias{Englisch}{english}
\DeclareLanguageAlias{EN}{english}
\DeclareLanguageAlias{en-US}{english}

\usepackage{braket}

\usepackage{graphicx}
\usepackage[export]{adjustbox}

\usepackage{xcolor}

\newcommand{\coParam}{\xi}

\newtheorem{theorem}{Theorem}
\begin{document}

\title{Designing pretty good state transfer via isospectral reductions}

\author{M. Röntgen}
\affiliation{%
	Zentrum für optische Quantentechnologien, Universität Hamburg, Luruper Chaussee 149, 22761 Hamburg, Germany
}%

\author{N. E. Palaiodimopoulos}%
\affiliation{%
	Department of Physics, University of Athens, 15771 Athens, Greece
}%

\author{C. V. Morfonios}%
\affiliation{%
	Zentrum für optische Quantentechnologien, Universität Hamburg, Luruper Chaussee 149, 22761 Hamburg, Germany
}%

\author{I. Brouzos}%
\affiliation{%
	Department of Physics, University of Athens, 15771 Athens, Greece
}%
\affiliation{%
	Le Mans University, LAUM UMR CNRS 6613, Av. O. Messiaen, 72085, Le Mans,
	France
}

\author{M. Pyzh}%
\affiliation{%
	Zentrum für optische Quantentechnologien, Universität Hamburg, Luruper Chaussee 149, 22761 Hamburg, Germany
}%

\author{F. K. Diakonos}%
\affiliation{%
	Department of Physics, University of Athens, 15771 Athens, Greece
}%

\author{P. Schmelcher}
\affiliation{%
	Zentrum für optische Quantentechnologien, Universität Hamburg, Luruper Chaussee 149, 22761 Hamburg, Germany
}%
\affiliation{%
	The Hamburg Centre for Ultrafast Imaging, Universität Hamburg, Luruper Chaussee 149, 22761 Hamburg, Germany
}%

\newcommand{\hamil}{\mathbf{H}}

\begin{abstract}

We present an algorithm to design networks that feature pretty good state transfer (PGST), which is of interest for high-fidelity transfer of information in quantum computing.
Realizations of PGST networks have so far mostly relied either on very special network geometries or imposed conditions such as transcendental on-site potentials.
However, it was recently shown [\emph{Eisenberg et al., arXiv:1804.01645}] that PGST generally arises when a network's eigenvectors and the factors $P_{\pm}$ of its characteristic polynomial $P$ fulfill certain conditions, where $P_{\pm}$ correspond to eigenvectors which have $\pm 1$ parity on the input and target sites.
We combine this result with the so-called isospectral reduction of a network
to obtain $P_{\pm}$ from a dimensionally reduced form of the Hamiltonian.
Equipped with the knowledge of the factors $P_{\pm}$, we show how a variety of setups can be equipped with PGST by proper tuning of $P_{\pm}$.
Having demonstrated a method of designing networks featuring pretty good state transfer of single site excitations, we further show how the obtained networks can be manipulated such that they allow for robust storage of qubits.
We hereby rely on the concept of compact localized states, which are eigenstates of a Hamiltonian localized on a small subdomain, and whose amplitudes completely vanish outside of this domain.
Such states are natural candidates for the storage of quantum information, and we show how certain Hamiltonians featuring pretty good state transfer of single site excitation can be equipped with compact localized states such that their transfer is made possible.
\end{abstract}

\maketitle

\section{Introduction}
The ability to reliably transfer information through a quantum system is of key importance in the quest towards quantum computers. 
One particularly appealing approach is that of \emph{perfect state transfer} (PST) \cite{Bose2003PRL91207901QuantumCommunicationUnmodulatedSpin,Christandl2004PRL92187902PerfectStateTransferQuantum,Christandl2005PRA71032312PerfectTransferArbitraryStates} of a given state -- usually a single site excitation of an $XY$-Hamiltonian -- from an input to a target site. 
What makes PST appealing is that it achieves perfect transfer fidelity (the portion of the final state at the desired site) $F = 1$ by simple time-evolution of the input excitation with the time-independent Hamiltonian.
From a realistic viewpoint, however, the strong requirement of unity fidelity is never met due to imperfections; it rather limits severely the flexibility in the design of quantum networks for state transfer.
A less restrictive alternative to PST is the concept of \emph{pretty good} \cite{Godsil2012DM312129StateTransferGraphs} (also called almost perfect \cite{Vinet2012PRA86052319AlmostPerfectStateTransfer}) \emph{state transfer} (PGST), where $F$ gets arbitrarily close to unity at a corresponding time: Specifically, for every $\epsilon > 0$ there is a time $t_{\epsilon}$ such that $F(t_{\epsilon}) > 1 - \epsilon$, where $F(t) = |\braket{\psi_{I}|exp(i \hamil t)|\psi_{F}}|^2$ (setting $\hbar = 1$) for a transfer from state $\ket{\psi_{I}}$ to state $\ket{\psi_{F}}$ at time $t$, with $\hamil$ denoting the Hamiltonian.
Clearly, PGST includes the case of PST and is therefore a broader concept.
Still, the design of PGST-Hamiltonians is challenging, since it usually requires information about the \emph{exact} eigenvalue spectrum.
So far, many approaches to PGST are therefore based on special Hamiltonian designs such as certain graph products \cite{Coutinho2016AQPrettyGoodStateTransfer,Fan2013LAIA4382346PrettyGoodStateTransfer,Pal2017EJC24PrettyGoodStateTransfer,vanBommel2016AQCompleteCharacterizationPrettyGood,Ackelsberg2016LAIA506154LaplacianStateTransferCoronas,Ackelsberg2016LAIA506154LaplacianStateTransferCoronas}.
A general and intuitive design mechanism of PGST-Hamiltonians is thus lacking.

Recently, progress in this direction has been made in Ref. \cite{Eisenberg2018DMPrettyGoodQuantumState}.
There, an approach is presented that achieves PGST between two sites $u$ and $v$ \emph{without} direct tuning of the eigenvalue spectrum.
The approach is based on Hamiltonians $\hamil$ which feature so-called cospectral sites $u$ and $v$ for a range of parameters.
In Hamiltonians with such cospectral sites $u$ and $v$, all eigenvectors can be chosen to have parity $\pm 1$ on $u$ and $v$.
Eisenberg et al. then show that PGST between $u$ and $v$ automatically arises if the factors $P_{\pm}$, which are related to eigenvectors which have non-vanishing amplitudes on $u$ and $v$ and additionally have $\pm 1$ parity on them, respectively, of the characteristic polynomial of $\hamil$ fulfill certain conditions.
The task of achieving PGST therefore boils down to proper tuning of the factors $P_{\pm}$.
In practice, though, obtaining these factors from the underlying Hamiltonian is not easy. In Ref. \cite{Eisenberg2018DMPrettyGoodQuantumState}, $P_{\pm}$ are (up to special cases involving symmetries or very small setups) not obtained, but indirect methods, which manipulate $\hamil$ such that $P_{\pm}$ are enforced to meet the desired properties, are presented.
An example of such a method is the addition of transcendental numbers to the values of certain on-site potentials of the Hamiltonian.
While elegant, this method limits the practical applicability, and the question arises whether other, more practical methods of designing PGST Hamiltonians exist.

In this work, we present such a method by pursuing an alternative road to PGST.
Namely, by directly obtaining the polynomials $P_{\pm}$ from an underlying symmetric Hamiltonian that features cospectral sites $u$ and $v$.
To this end, we combine the mathematical relations underlying the works in Ref. \cite{Eisenberg2018DMPrettyGoodQuantumState,Kempton2017QIP16210PrettyGoodQuantumState} with the theory of isospectral reductions \cite{Bunimovich2011N25211IsospectralGraphTransformationsSpectral,Bunimovich2012C22033118IsospectralCompressionOtherUseful,Bunimovich2014IsospectralTransformationsNewApproach,VasquezFernandoGuevara2014NLAwA22145PseudospectraIsospectrallyReducedMatrices,Duarte2015LAIA474110EigenvectorsIsospectralGraphTransformations,Smith2019PA514855HiddenSymmetriesRealTheoretical,Kempton2019AMCharacterizingCospectralVerticesIsospectral}.
Isospectral reduction is a method to reduce the size of a given Hamiltonian while keeping a large amount of information on its eigenvalues and eigenvectors.
We utilize the isospectral reduction to ``compress'' only the relevant spectral information for the problem at hand by building upon the very recent results of Ref. \cite{Kempton2019AMCharacterizingCospectralVerticesIsospectral}.
These results put strong constraints on the structure of the isospectral reduction of a Hamiltonian that features cospectral sites.
We use these structural constraints to extract the $P_{\pm}$ from the isospectral reduction of $\hamil$.
Equipped with $P_{\pm}$, we show how this allows for a convenient and powerful algorithm for designing Hamiltonians featuring PGST by properly tuning $P_{\pm}$ whilst maintaining the cospectrality of $u$ and $v$.
In order to be self-contained, we also collect known facts from the literature and condense them into a detailed method to generate Hamiltonians that feature cospectral sites $u$ and $v$.

Interestingly, this cospectrality is often accounted for by \emph{spatial local symmetries}, i.e., symmetries, which are only valid in spatial subdomains of the whole system.
Usually, the signatures of such local symmetries are only indirectly encoded into so-called non-local currents, as has been shown in Refs. \cite{Kalozoumis2013PRA88033857LocalSymmetriesPerfectTransmission,Kalozoumis2014PRL113050403InvariantsBrokenDiscreteSymmetries,Kalozoumis2015AP362684InvariantCurrentsScatteringLocally,Zampetakis2016JPA49195304InvariantCurrentApproachWave,Rontgen2017AP380135NonlocalCurrentsStructureEigenstates,Morfonios2017AP385623NonlocalDiscreteContinuityInvariant,Morfonios2017AP385623NonlocalDiscreteContinuityInvariant}.
On the contrary, the impact of the underlying local symmetries is directly visible in setups featuring cospectral vertices $u$ and $v$, where all eigenvectors are (in the case of degeneracies, can be chosen to be) locally parity symmetric on these sites.
It would thus be interesting to analyze cospectral Hamiltonians within the framework developed in those works.

Having demonstrated how to design networks capable of PGST of single site excitations, we show how these networks can be modified to allow for robust storage of qubits.
To this end, we slightly modify these networks, thereby equipping them with so-called \emph{compact localized states}.
Such states are eigenstates of the underlying Hamiltonian \cite{Rontgen2018PRB97035161CompactLocalizedStatesFlat,Leykam2018AP31473052ArtificialFlatBandSystems}, and are perfectly localized on a finite number of sites.
They are thus ideally suited for the storage of qubits, and we show how, after equipping networks with compact localized states, these can also be pretty well transferred.

This work is structured as follows.
We first define the necessary and sufficient conditions for the realization of PGST in \cref{sec:necAndSufCond}.
We then investigate the necessary condition, namely strong cospectrality, which is a stronger version of cospectrality, in more detail in \cref{sec:cospectrality}, and the connection of this property to symmetries in \cref{sec:cospecSymmetries}.
Our treatment of strong cospectrality is completed in \cref{sec:designingCospectral}, where we show how Hamiltonians with this property can be designed.
In \cref{sec:relationMinimalPolynomials}, we introduce isospectral reductions, and show how they can be harnessed to extract the polynomials $P_{\pm}$.
In \cref{sec:Application}, we use this method to construct an algorithm for the design of graphs featuring PGST. This algorithm represents the main novelty and also one of the two highlights of this work.
We apply the algorithm to a simple example in \cref{sec:Example}.
In \cref{sec:CLSTransfer}, we present the necessary modifications needed for PGST of compact localized states.
Finally, we conclude our work in \cref{sec:conclusions}.

\section{Theory: Pretty good state transfer of single-site excitations} \label{sec:MathematicalFoundations}
Throughout this work, we will consider setups described by symmetric Hamiltonians of the form
\begin{equation} \label{eq:Hamiltonian}
	\hamil = \sum_{i} E_{i} \ket{i}\bra{i} + \sum_{<i,j>}^{} h_{i,j} \ket{i}\bra{j}
\end{equation}
with real on-site potentials $E_{i}$ and couplings $h_{i,j} = h_{j,i}$, where the sum in \cref{eq:Hamiltonian} runs over all interconnected sites $i$ and $j$.
We use bold-faced script for both vectors and matrices.
The Hamiltonian given by \cref{eq:Hamiltonian} can be represented, for example, by coupled waveguide arrays \cite{Garanovich2012PR5181LightPropagationLocalizationModulated,Szameit2012DiscreteOpticsFemtosecondLaser}.
However, in the context of quantum computers, a natural choice are spin networks, where each site represents a spin-$1/2$ qubit (measured up or down).
The Heisenberg XX interaction Hamiltonian then reduces to the simple description \cref{eq:Hamiltonian} within the subspace of one excitation (1 spin up and all others down) \cite{Bose2007CP4813QuantumCommunicationSpinChain}.

In the course of this work, we will often depict $\hamil$ as a graph, i.e., as a collection of vertices and edges connecting them.
The adjacency matrix of the graph then equals $\hamil$.
As this establishes a one-to-one relation between graphs and the underlying $\hamil$, and we will use these two terms interchangeably. Likewise, we will use the terms ``site'' and ``vertex'' interchangeably throughout this work.

In the following, we will comment on the conditions for PGST, and show how it can be achieved.
In order to help the reader in comprehending the different aspects involved, we visualize in \cref{fig:guidingFigure} the main mathematical background of our method to achieve PGST.
The overview \cref{fig:guidingFigure} contains all core mathematical theorems used in this work in compact form.
We stress that the style of presentation at this point aims at being self-contained, thereby transferring insights from graph theory to a broader physics community in \cref{sec:cospectrality,sec:cospecSymmetries,sec:designingCospectral}.
\Cref{sec:relationMinimalPolynomials,sec:Application} contain, along with the results provided in \cref{sec:CLSTransfer}, the highlights and main novelties of this work.

\begin{figure*} 
	\centering
	\includegraphics[max size={\textwidth}{\textheight}]{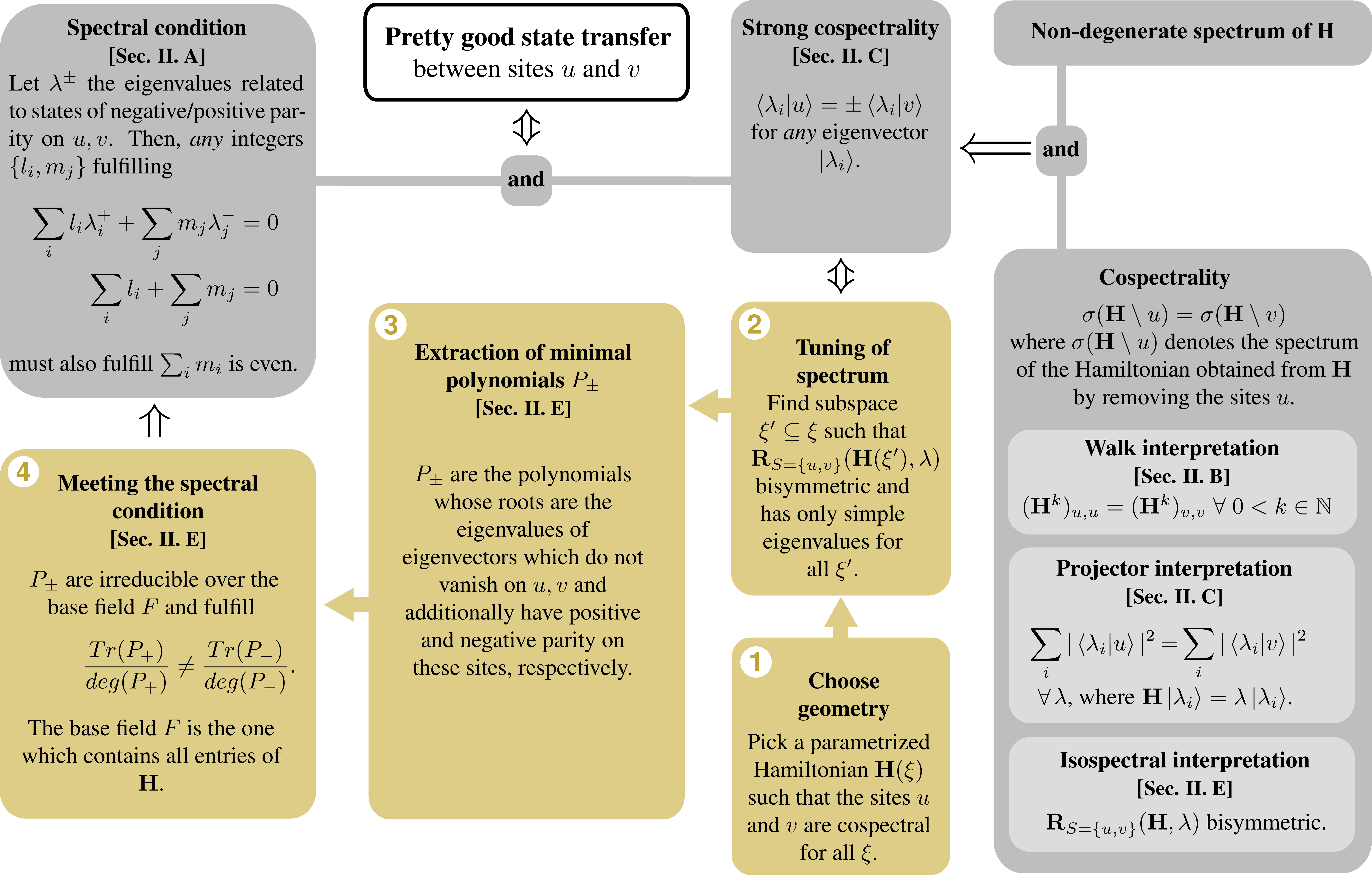}
	\caption{Grey part: The connection between main theorems that lead to realizations of PGST. Double lines with arrows denote mathematical relations $\Leftarrow,\Rightarrow,\Leftrightarrow$. Yellow/golden part: Essential steps (enumerated $1$ to $4$) of the algorithm for the design of PGST-Hamiltonians [presented in \cref{sec:Application}].
	The isospectral reduction $\mathbf{R_{S}(\hamil,\lambda)}$ is defined in \cref{eq:isospectralReduction}.
	}
	\label{fig:guidingFigure}
\end{figure*}

\subsection{Necessary and sufficient conditions} \label{sec:necAndSufCond}
In order to support PGST between two sites $u$ and $v$, a Hamiltonian $\hamil$ has to fulfill the following two conditions, whose combination is necessary and sufficient: (i) The two sites $u$ and $v$ must be \emph{strongly cospectral} (see the following subsection) and (ii) its spectrum must fulfill the following condition \cite{Banchi2017JMP58032202PrettyGoodStateTransfer}:
Any integers $\{l_{i},m_{j}\}$ which fulfill
\begin{align} \label{eq:PGSTCond1}
\sum_{i} l_{i} \lambda^{+}_{i} + \sum_{j} m_{j} \lambda^{-}_{j} &= 0\\
\sum_{i} l_{i} + \sum_{j} m_{j} &= 0 \label{eq:PGSTCond2}
\end{align}
must also fulfill
\begin{equation} \label{eq:PGSTCond3}
\sum_{i} m_{i} \; \text{is even} .
\end{equation}
Here, $\lambda^{+}_{i}, \lambda^{-}_{j}$ are the eigenvalues associated to eigenvectors $\ket{\psi_{i}^{+}}, \ket{\psi_{j}^{-}}$ of $\hamil$ that fulfill
\begin{align*}
	\braket{\psi_{i}^{+}|u} &= +\braket{\psi_{i}^{+}|v} \ne 0 \\
	\braket{\psi_{j}^{-}|u} &= -\braket{\psi_{j}^{-}|v} \ne 0
\end{align*}
where $\ket{u},\ket{v}$ describe single-site excitations of sites $u$ and $v$, respectively.

Note that there is always at least one set of integers $\{l_{i},m_{j}\}$ fulfilling \cref{eq:PGSTCond1,eq:PGSTCond2}, namely, the trivial choice $l_{i} = m_{j} = 0 \; \forall \; i,j$, which also fulfills \cref{eq:PGSTCond3}.
In certain cases, this trivial choice is also the \emph{only one} fulfilling \cref{eq:PGSTCond1,eq:PGSTCond2}.
An example is the case where there are only two eigenvalues, $\lambda^{+} = \sqrt{2}, \lambda^{-} = \sqrt{3}$.
Then
\begin{equation*}
	l \sqrt{2} + m \sqrt{3} = 0,
\end{equation*}
for integers $l,m$, can only be fulfilled when $l = m = 0$. The setup would thus feature PGST between $u$ and $v$.

\subsection{Geometric interpretation of cospectrality} \label{sec:cospectrality}

\begin{figure} 
	\centering
	\includegraphics[max size={\columnwidth}{\textheight}]{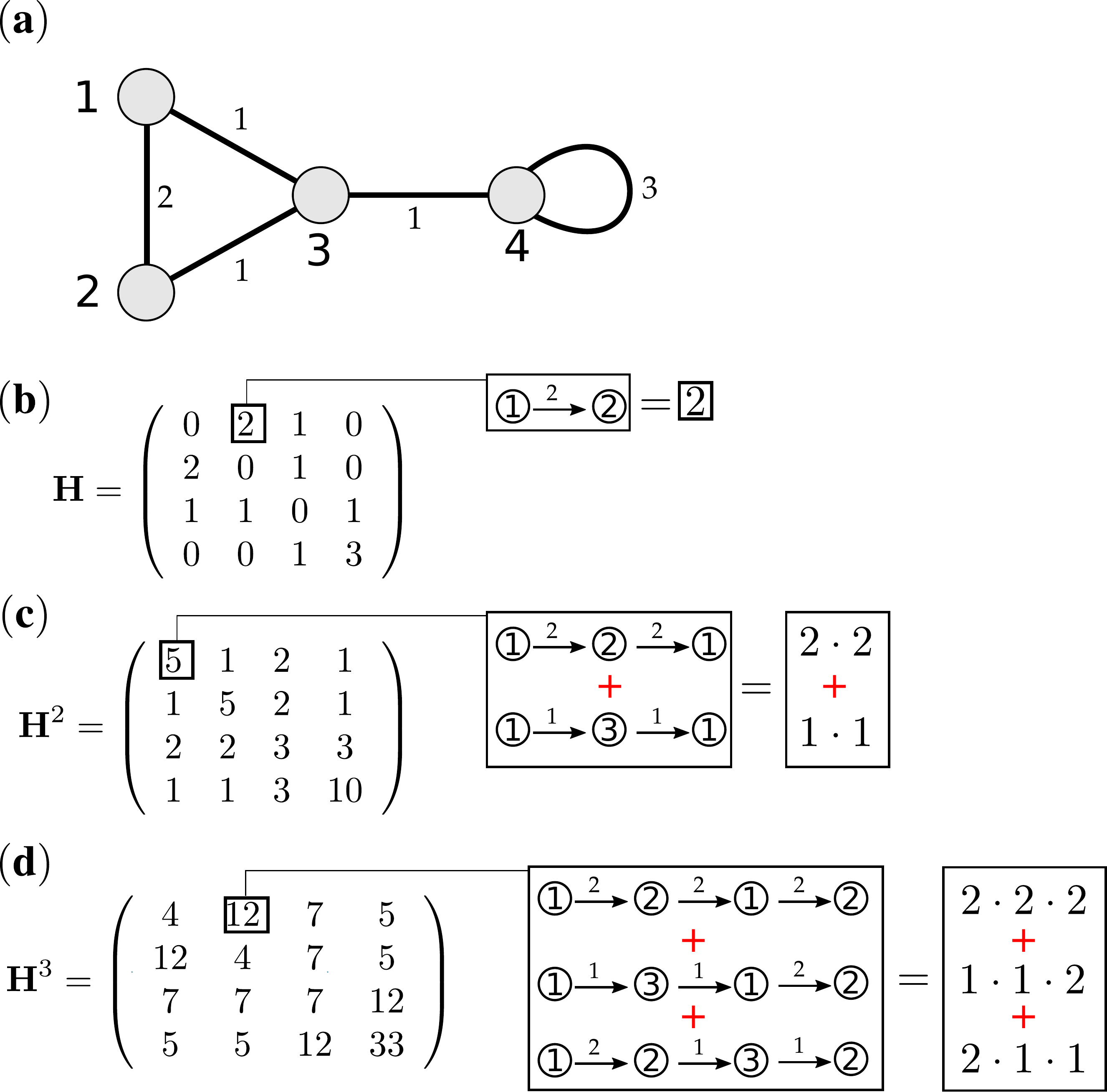}
	\caption{The entries of matrix powers of any matrix can be interpreted in terms of walks. \textbf{(a)} shows a graph, described by the matrix $\hamil$ denoted in \textbf{(b)}. We have here used the convention from graph theory that diagonal elements $\hamil_{i,i}$ are plotted as links from $i$ to itself. \textbf{(b -- d)} provide help on how to interpret the entries of powers $\hamil^{k}$ (see text for details).
	}
	\label{fig:WalkInterpretation}
\end{figure}

\begin{figure} 
	\centering
	\includegraphics[max size={\columnwidth}{\textheight}]{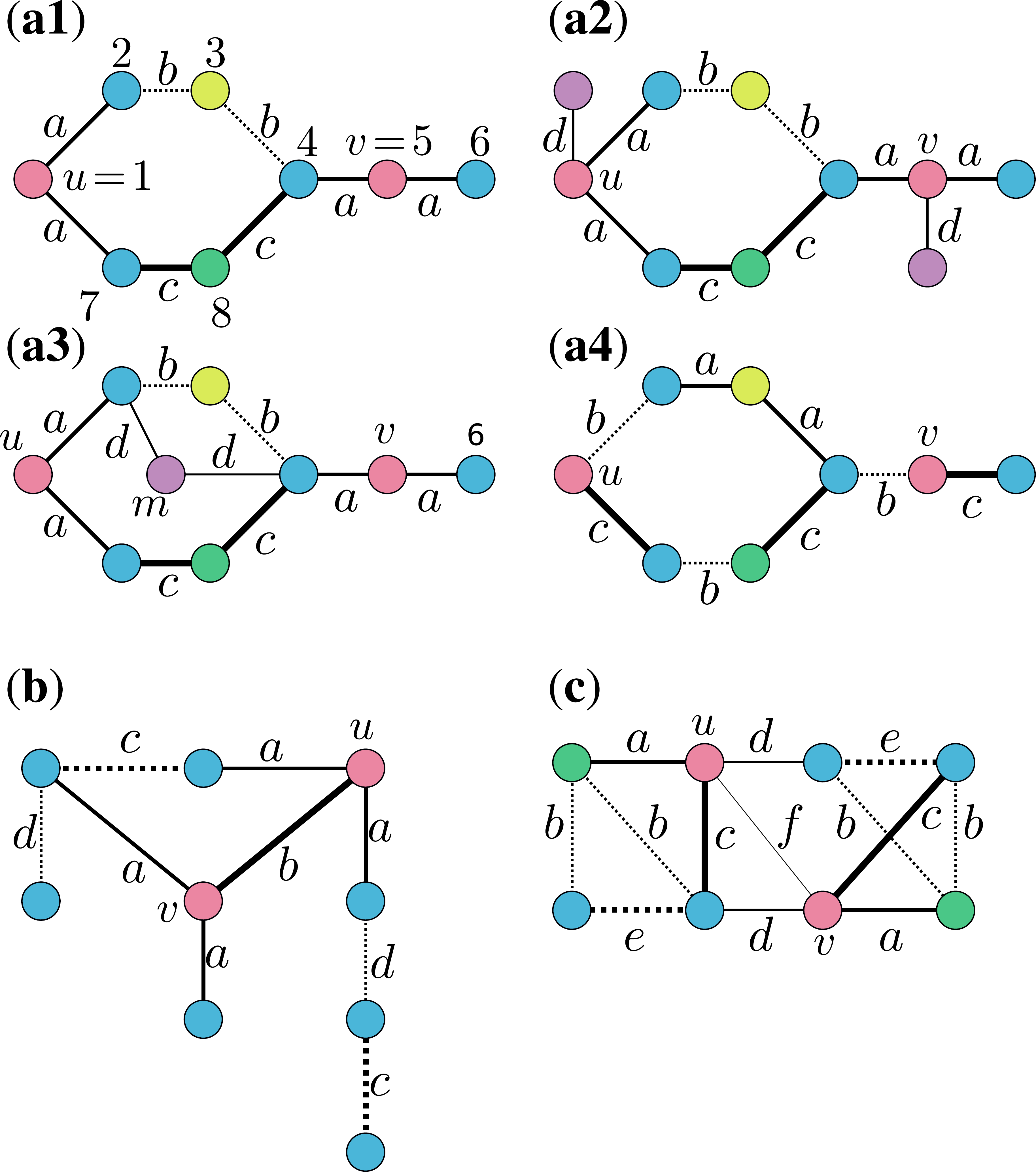}
	\caption{A collection of graphs featuring cospectral vertices. In every graph, the two red vertices are cospectral, provided that the equality of certain couplings (denoted by different line styles and labels) is kept, and that vertices of equal color share the same, arbitrarily valued on-site potential.
	}
	\label{fig:latentSymmetriesParameters}
\end{figure}

A necessary condition for PGST is that $u$ and $v$ are \emph{strongly cospectral} \cite{Banchi2017JMP58032202PrettyGoodStateTransfer}.
As strongly cospectral vertices are also cospectral \cite{Godsil2017AMStronglyCospectralVertices} [see \cref{fig:guidingFigure}], we will first investigate and understand this weaker property before turning to its stronger version.
Two vertices $u,v$ are said to be cospectral if $\sigma(\hamil\setminus u) = \sigma(\hamil \setminus v)$, where $\sigma(\hamil)$ denotes the eigenvalue spectrum of $\hamil$, and $\hamil \setminus u$ denotes the Hamiltonian obtained from $\hamil$ by deleting the $u$th row and column.
For our purpose, it is easier to rely on an equivalent condition \cite{Godsil2017AMStronglyCospectralVertices} in terms of the diagonal entries of powers of $\hamil$.
Namely, $u$ and $v$ are cospectral if and only if
\begin{equation} \label{eq:cospectralityConditionMatrixPowers}
	(\hamil^{k})_{u,u} = (\hamil^{k})_{v,v}
\end{equation}
for all non-negative integers $k < N$, where $\hamil \in \mathbb{R}^{N \times N}$.
As a side remark, we note that \cref{eq:cospectralityConditionMatrixPowers} automatically holds for all $k > 0$ provided that it holds for $0 < k < N$.
This is due to the Cayley-Hamilton theorem, which states that every matrix power $\hamil^{k\ge N}$ can be expanded in terms of smaller powers $\hamil^{k < N}$, i.e.,
\begin{equation*}
\hamil^{k\ge N} = \sum_{i=0}^{N-1} a_{i}^{(k)} \hamil^{i}
\end{equation*}
with $a_{i}^{(k)}$ being the expansion coefficients.
These coefficients are scalars, and therefore \cref{eq:cospectralityConditionMatrixPowers} must hold for all non-negative integers $k$, provided that it holds for $0 \le k<N$.

While well-known in graph theory, it is perhaps surprising to many physicists that the entries of $\hamil^k$ (for integer $k>0$) possess a convenient interpretation.
To this end, we interpret the Hamiltonian matrix as a graph, i.e., as a network of vertices $V_{i}$ connected to each other by weighted edges $e_{i,j} = \{V_{i},V_{j}\}$, with weight $w(e_{i,j}) = \hamil_{i,j}$.
This is exemplarily done in \cref{fig:WalkInterpretation} (\textbf{a}) for the Hamiltonian given in \cref{fig:WalkInterpretation} (\textbf{b}).
In this picture of representing $\hamil$ as a graph, every matrix element $\hamil_{i,j}\ne 0$ is connected to an edge between vertices $i$ and $j$; and in particular, non-vanishing diagonal elements $\hamil_{i,i}\ne 0$ refer to a link from site $i$ to itself, with weight given by $w(e_{i,i}) = \hamil_{i,i}$.

Now that we have interpreted the entries of $\hamil$ in terms of edges, we show how entries of higher-order powers $\hamil^{k>1}$ can be interpreted in terms of \emph{walks}.
A walk can be thought of a route through the graph from one vertex to another by walking along the edges connecting neighboring vertices.
Mathematically, it is defined as an alternating sequence of vertices and edges, where each edge must connect its precursor vertex to its successor.
For example, in \cref{fig:WalkInterpretation} \textbf{(a)}, a walk of length $2$ from vertex $1$ to $4$ would be the sequence $p = \{V_{1}, e_{1,3}, V_{3}, e_{3,4}, V_{4}\}$.
In order to interpret the entries of $\hamil^k$, we note that, just as each edge $e_{i,j}$ can be given a weight $w(e_{i,j})$, we can also give each walk a weight by \emph{multiplying the weights of all edges occurring within this walk}.
Thus, the weight of the walk $p = \{V_{1}, e_{1,3}, V_{3}, e_{3,4}, V_{4}\}$ would be $w(p) = w(e_{1,3}) \cdot w(e_{3,4}) = 1 \cdot 1$. Equipped with these definitions, one can show that [a proof is provided in \cref{app:walkInterpretationProof}]
\begin{equation} \label{eq:walkInterpretation}
	(\hamil^{k>0})_{a,b} = \sum_{p} w\left( p_{a,b}^{(k)} \right)
\end{equation}
where $p_{a,b}^{(k)}$ denotes one possible walk of length $k$ between vertices $a$ and $b$, and the sum is over all such walks.
In other words, the value of the matrix element $(\hamil^{k})_{i,j}$ is equal to the sum of weights of all walks of length $k$ between vertices $i$ and $j$.
In \Cref{fig:WalkInterpretation} (\textbf{b -- d}), we have visualized this interpretation of walks, and have also explicitly given the integer powers $\hamil^{k<N}$.

We now connect the interpretation of matrix elements of $\hamil^{k}$ in terms of walks to the cospectrality of two vertices $u$ and $v$.
As we have seen above, these are cospectral if and only if \cref{eq:cospectralityConditionMatrixPowers} is fulfilled for all integer $k <N$, with $N$ being the number of sites contained in $\hamil$.
Now, by interpreting the entries of $\hamil^{k}$ in terms of walks, the cospectrality of $u$ and $v$ can therefore be determined in a simple and straightforward manner.
Namely, by evaluating all walks of length $k$ that go from $u$ onto itself, and those that go from $v$ onto itself, and comparing the respective sum of weights, order by order in $k<N$.
Thus, in \cref{fig:WalkInterpretation}, $(\hamil^{k})_{1,1} = (\hamil^{k})_{2,2}$ for $k < 4$ [and, by the Cayley-Hamilton theorem, also for all integer $k>0$], which makes the sites $u=1$ and $v=2$ cospectral.
Alternatively, one can also rely on the statement that two sites $u$ and $v$ are cospectral if and only if $\sigma(\hamil\setminus u) = \sigma(\hamil \setminus v)$.
As the graph of $\hamil \setminus 1$ is identical to that of $\hamil \setminus 2$, their spectra $\sigma(\hamil \setminus 1) = \sigma(\hamil \setminus 2)$ are trivially identical, and the sites $u=1$ and $v=2$ are therefore cospectral.

We show a collection of cospectral graphs in \cref{fig:latentSymmetriesParameters}.
In every graph, the two red vertices (labeled $u$ and $v$) are cospectral, provided that any two couplings denoted by the same label and line style are identical, and that any two vertices sharing the same color also have identical on-site potential.
Let us now investigate these graphs in more detail.
By comparing different variations, it can be seen that certain changes do not break the cospectrality of two vertices.
For example, in \cref{fig:latentSymmetriesParameters} (\textbf{a2}), we have modified the graph from \cref{fig:latentSymmetriesParameters} (\textbf{a1}) by identically coupling each red vertex to an additional purple vertex.
In \cref{fig:latentSymmetriesParameters} (\textbf{a3}), we have modified the graph of \cref{fig:latentSymmetriesParameters} (\textbf{a1}) by inserting the ``central'' vertex $m$.
We term this vertex central since it can be reached from sites $1$ and $5$ by a walk comprising two steps, and the corresponding weights of these two walks are identical.
To understand why inserting this vertex does not break the cospectrality, one only needs to investigate the influence of this change by comparing the diagonal matrix elements $(\hamil_{\text{B}}^{k})_{S,S}$ and $(\hamil^{k})_{S,S}$. Here, $\hamil_{\text{B}}$ and $\hamil$ describe the setup of \cref{fig:latentSymmetriesParameters} (\textbf{a1}) and (\textbf{a3}), respectively, and $S = \{u,v\}$ label the two red sites.
Before the change, $u$ and $v$ were cospectral, so that $(\hamil_{\text{B}}^{k})_{u,u} = (\hamil_{\text{B}}^{k})_{v,v} \; \forall \; k$.
Thus, to understand why the cospectrality is kept, we only need to look at the differences $(\Delta \hamil^{k})_{S,S} = (\hamil^{k})_{S,S} - (\hamil_{\text{B}}^{k})_{S,S}$ caused by inserting the new vertex.
Though tedious, it is a straightforward task to show that $(\Delta \hamil^{k})_{1,1} = (\Delta \hamil^{k})_{5,5}$ for $k<9$, and, by the Cayley-Hamilton theorem, therefore for all $k$.
The addition of vertex $m$ does thus, at each order $k$, add an equally valued sum of weights of walks from site $1$ to itself compared to those from site $5$ to itself.
For this reason, its addition does not change the cospectrality of $u$ and $v$.

\subsection{Strong cospectrality and the impact of symmetries} \label{sec:cospecSymmetries}
As we have seen above, cospectrality is linked to the geometric and spectral properties of a graph.
It is likewise linked to properties of the graph's eigenstates.
Indeed, it can be shown that two sites $u,v$ are cospectral if and only if \cite{Godsil2017AMStronglyCospectralVertices}, for $\hamil \ket{\lambda_{i}} = \lambda \ket{\lambda_{i}}$,
\begin{equation} \label{eq:cospectralProjectors}
	\sum_{i}^{} |\braket{\lambda_{i}|u}|^2 = \sum_{i}^{} |\braket{\lambda_{i}|v}|^2
\end{equation}
is fulfilled for all $\lambda$.
In words, $u$ and $v$ are cospectral if and only if, within each degenerate subspace, the sum of squares of absolute values of projections on sites $u$ is equal to that of projections on site $v$.

If $u$ and $v$ are cospectral and additionally \cite{Eisenberg2018DMPrettyGoodQuantumState}
\begin{equation} \label{eq:StronglyCospectralProjectors}
\braket{\tilde{\lambda}_{i} |u} = \pm \braket{\tilde{\lambda}_{i} |v}
\end{equation}
for \emph{any} superposition $\ket{\tilde{\lambda}_{i}} = \sum_{j} c_{j} \ket{\lambda_{j}}$ of degenerate states $\ket{\lambda_{j}}$, then $u$ and $v$ are said to be \emph{strongly cospectral}.
Therefore, strong cospectrality implies cospectrality, but that the reverse is not necessarily true.
Unlike cospectrality, which can be readily interpreted and tested for in terms of walks, we are not aware of an easy, i.e., without computing the determinant or the eigenstates of the graph, method to test whether a given general graph is strongly cospectral or not. As the field of cospectral vertices is quite young, there is hope that this may change in the future, and we refer the interested reader to Ref. \cite{Godsil2017AMStronglyCospectralVertices} for further information on the fascinating field of strongly cospectral vertices.

With the above statements in mind, let us now investigate the symmetries of cospectral graphs. To this end, we compare the graphs shown in \cref{fig:latentSymmetriesParameters} to the one shown in \cref{fig:WalkInterpretation} (\textbf{a}). The latter graph has the special property that the underlying Hamiltonian is invariant under the permutation of vertices $1$ and $2$ and therefore commutes with the corresponding permutation operator.
As is well-known, such a symmetry has a drastic impact: The eigenstates are (or, in case of degeneracies, can be chosen to have) parity $\pm 1$ with respect to a flip of sites $u$ and $v$. Thus, they fulfill \cref{eq:cospectralProjectors}, so that $u$ and $v$ are cospectral.
Provided that states of negative and positive parity are non-degenerate to each other, they additionally fulfill \cref{eq:StronglyCospectralProjectors}, so that $u$ and $v$ are even strongly cospectral.
While the fact that a permutation symmetry of $u$ and $v$ leads to their (strong) cospectrality should be no surprise, things change when inspecting the graphs shown in \cref{fig:latentSymmetriesParameters}. While they are indeed all cospectral, none of them is invariant under any non-trivial permutation of vertices.
In other words, the underlying Hamiltonian does not commute with the corresponding permutation matrices. However, due to cospectrality their eigenstates fulfill \emph{the same equation} \cref{eq:cospectralProjectors} [and, depending on degeneracies, also \cref{eq:StronglyCospectralProjectors}] as they would do in the presence of a permutation symmetry.
For this reason, graphs (or, just as well, matrices) that lack direct symmetries, but whose eigenstates fulfill \cref{eq:cospectralProjectors,eq:StronglyCospectralProjectors} were recently termed \emph{latently symmetric} \cite{Smith2019PA514855HiddenSymmetriesRealTheoretical,Kempton2019AMCharacterizingCospectralVerticesIsospectral}.
However, although these symmetries may indeed seem hidden, we would like to mention here that all the graphs shown in \cref{fig:latentSymmetriesParameters} \emph{indeed} feature \emph{local symmetries}, i.e., symmetries within subdomains of the system, such that the underlying symmetry operations commute with the Hamiltonian of the subsystem, but not with that of the complete one. An example is the subsystem of sites $2,3,4$ in \cref{fig:latentSymmetriesParameters} (\textbf{a1}), which is invariant under the permutation of sites $2$ and $4$.
Given the high number of local symmetries in latently symmetric setups, it would be interesting to investigate such systems under the recently established framework of local symmetries \cite{Kalozoumis2015AP362684InvariantCurrentsScatteringLocally,Kalozoumis2014PRL113050403InvariantsBrokenDiscreteSymmetries,Kalozoumis2013PRA88033857LocalSymmetriesPerfectTransmission,Morfonios2017AP385623NonlocalDiscreteContinuityInvariant,Morfonios2017AP385623NonlocalDiscreteContinuityInvariant,Rontgen2017AP380135NonlocalCurrentsStructureEigenstates,Zampetakis2016JPA49195304InvariantCurrentApproachWave}, which provides dedicated tools for the analysis of such setups.

As a concluding remark, we note that there are still many questions open regarding the connection between local symmetries of a Hamiltonian $\hamil$ and the cospectrality of two sites $u$ and $v$ of $\hamil$.
Given such a Hamiltonian, it is clear that the subsystem $\hamil_{SS} \in \mathbb{R}^{2 \times 2}$ with $S = \{u,v\}$ is invariant under the exchange of $u$ and $v$, since $\hamil$ is symmetric and cospectrality of $u$ and $v$ implies that $\hamil_{u,u} = \hamil_{v,v}$.
However, it is yet unknown whether and to which amount $\hamil$ must \emph{necessarily} feature more (i.e., apart from that of $\hamil_{SS}$) local symmetries in order to allow for the cospectrality of $u$ and $v$.
Although the question about the necessity of local symmetries for cospectrality is thus still open, local symmetries often naturally appear during the process of designing networks with cospectral vertices, as we will see in the following section.

\subsection{Designing graphs featuring strongly cospectral vertices} \label{sec:designingCospectral}
\begin{figure*} 
	\centering
	\includegraphics[max size={0.7\textwidth}{\textheight}]{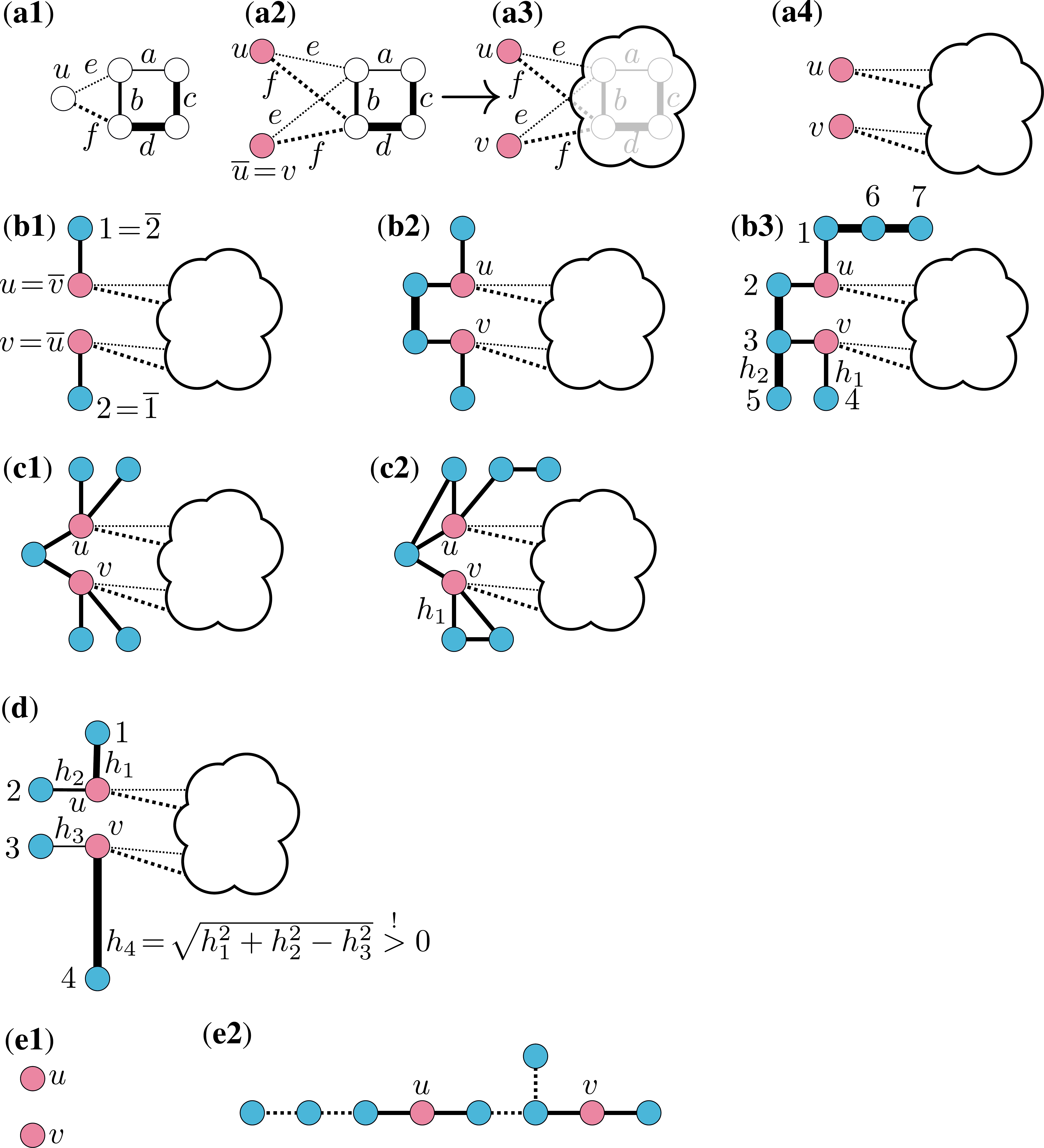}
	\caption{\textbf{(a1 -- a2)}: Creation of a graph [as the one shown in \textbf{(a2)}] featuring cospectral vertices from an arbitrary graph [as the one shown in \textbf{(a1)}] by symmetrization of the site $u$. In \textbf{(a3)}, we divide the given graph featuring cospectral sites $u$ and $v$ into two parts: The two red sites denoting $u,v$, and the remainder of the graph, denoted by a cloud.
		In the remainder of this figure, the combination of the two red sites with a cloud denote an arbitrary subsystem featuring cospectral vertices $u$ and $v$. In \textbf{(a4 -- e2)}, we present a number of operations on such a general graph which preserve the cospectrality. The validity of these operations is proven in \cref{app:CospectralityProofs}.
	}
	\label{fig:cospectralityDesign}
\end{figure*}

Let us now briefly recapitulate the above. We have seen that a necessary condition for PGST from $u$ to $v$ is strong cospectrality of these two vertices. This implies that $u$ and $v$ are cospectral, and we demonstrated that whether $u$ and $v$ are cospectral can be easily determined by testing whether $(\hamil^{k})_{u,u} = (\hamil^{k})_{v,v} \; \forall\; k < N$.
We then showed how these matrix entries can in turn be determined by summing up the respective weights of all possible walks of length $k$ from sites $u$ and $v$ to themselves.
In the following, we will show how one can \emph{design} graphs featuring strongly cospectral vertices $u$ and $v$.
We will start with a simpler problem, namely, the design of graphs with cospectral vertices $u$ and $v$, and then show how strong cospectrality can be achieved.

A convenient way to create a graph with cospectral vertices is to take any graph, and replicate and symmetrize one of its sites, $u$, as shown in \cref{fig:cospectralityDesign} (\textbf{a1} -- \textbf{a4}). This symmetrization then automatically yields the cospectrality of site $u$ and its symmetry partner $\overline{u} = v$ (see figure). This is due to the fact that the underlying Hamiltonian is invariant under an exchange of $u$ and $v$, which can easily be shown to imply cospectrality of these vertices.

Having seen how a graph can be changed to feature cospectral vertices $u$ and $v$, let us now show some modifications of this graph which \emph{keep the cospectrality}.
The procedure is shown in \cref{fig:cospectralityDesign} \textbf{(b1 -- e2)}, but let us elaborate more on its subfigures \textbf{(a1 -- a3)} first.
\Cref{fig:cospectralityDesign} (\textbf{a1} - \textbf{a2}) show the above symmetrization procedure for a simple example setup of five sites.
In \cref{fig:cospectralityDesign} \textbf{(a3)}, the logic underlying the subfigures \textbf{(b1 -- e2)} is shown.
Namely, the cloud incorporates a subgraph which must be chosen such that the composite graph, consisting of this subgraph and the two red vertices, features cospectrality of these two red vertices. This subgraph can consist of the four vertices as shown in \cref{fig:cospectralityDesign} \textbf{(a3)}, but can likewise be an arbitrarily complicated structure as long as the two red vertices are cospectral in the composite structure.
In \cref{fig:cospectralityDesign} \textbf{(a4)}, we show a first composite structure.
We then change it by a series of modifications in \cref{fig:cospectralityDesign} \textbf{(b1 -- b3)}, \cref{fig:cospectralityDesign} \textbf{(c1 -- c2)}, \cref{fig:cospectralityDesign} \textbf{(d)} and \cref{fig:cospectralityDesign} \textbf{(e1 -- e2)}.
Each of these modifications keeps the cospectrality of the two red vertices, as we prove in \cref{app:CospectralityProofs}.

These modifications can be divided into two classes: Those where the subsystem $\hamil_{\text{BR}}$ consisting of the red and blue sites is reflection symmetric about the horizontal axis, so that its sites $i$ are transformed as $i \leftrightarrow \overline{i}$ [compare \cref{fig:cospectralityDesign} \textbf{(b1)}] and those where there is no such symmetry.
The graphs shown in \cref{fig:cospectralityDesign} \textbf{(b1,b2,c1,e1)} belong to the first class.
The principle underlying this class of modifications is that, under the symmetry operation of a reflection of $\hamil_{\text{BR}}$ about the horizontal axis, the two sites $u$ and $v$ are mapped onto each other.
Due to this symmetry, they are trivially cospectral within $\hamil_{\text{BR}}$.
As we show in \cref{app:CospectralityProofs}, their cospectrality is preserved also in the full system, where $\hamil_{\text{BR}}$ is connected to the cloud.
The graphs shown in \cref{fig:cospectralityDesign} \textbf{(b3,c2,d,e2)} belong to the second class.
In these setups, the corresponding subgraph $\hamil_{\text{BR}}$ is no longer symmetric at all.
Nevertheless, the sites $u$ and $v$ are cospectral, and their cospectrality can readily be understood by evaluating the powers of $\hamil_{\text{BR}}$, as was done in \cref{sec:cospectrality}.

We thus demonstrate a set of examples which allow to design a variety of cospectral graphs from simpler structures such as the one shown in \cref{fig:cospectralityDesign} \textbf{(a2)}.
In \cref{fig:cospectralityDesign} (\textbf{e1 -- e2}), we apply the modifications done in \cref{fig:cospectralityDesign} \textbf{(b1 -- b3)} to two isolated sites (the red vertices).
We thereby create the iconic graph that is shown in the first paper on cospectral vertices by Schwenk \cite{Schwenk1973PTAAC257AlmostAllTreesAre} and is also depicted in many publications related to cospectrality, for example in Refs. \cite{Godsil2017AMStronglyCospectralVertices,Kempton2019AMCharacterizingCospectralVerticesIsospectral,Eisenberg2018DMPrettyGoodQuantumState}.

We stress that the operations presented above are certainly only a subset of valid modifications that keep cospectrality.
As we noted above, the study of cospectral vertices is still an emerging field, and we expect that there will be more construction principles found in the future.
To help the reader and to spread the understanding of graphs with cospectral vertices, we developed a graphical MATLAB tool that allows to design graphs and check for cospectrality of vertices in an intuitive and fast way.
This tool is available upon request from the authors.

Let us now come back to a statement about local symmetries, made in the last paragraph of \cref{sec:cospecSymmetries}.
There, we stated that local symmetries often occur naturally during the process of designing Hamiltonians featuring cospectral sites.
We can now support this statement by looking at the graphs depicted in \cref{fig:cospectralityDesign} \textbf{(b3)} and \textbf{(c2)}, both of which feature local symmetries.
In \cref{fig:cospectralityDesign} \textbf{(b3)}, the subgraph consisting of the sites $u,v,1,2,3,4$ is invariant under the permutation $1 \leftrightarrow 4, 2 \leftrightarrow 3, u \leftrightarrow v$.
This local symmetry is caused by the way \cref{fig:cospectralityDesign} \textbf{(b3)} was constructed.
Namely, by first making symmetric changes [performed in \cref{fig:cospectralityDesign} \textbf{(b1 -- b2)}] to the initial setup of \cref{fig:cospectralityDesign} \textbf{(a4)}, and breaking them afterwards by performing another change, as done in \cref{fig:cospectralityDesign} \textbf{(b3)}.
The underlying symmetries present in \cref{fig:cospectralityDesign} \textbf{(b1 -- b2)}] are then rendered to be local symmetries in \cref{fig:cospectralityDesign} \textbf{(b3)}].
A similar reasoning can be done for the setup depicted in \cref{fig:cospectralityDesign} \textbf{(c2)}.
Local symmetries can also occur accidentally, as we now show.
To this end, we note that the graph in \cref{fig:cospectralityDesign} \textbf{(d)} was designed on purpose such that $u$ and $v$ are cospectral for \emph{arbitrary} $h_{1},h_{2},h_{3},h_{4} > 0$.
In particular, the $h_{i}$ can be chosen asymmetrically, i.e., such that no any two couplings are identical.
Yet, as a byproduct of this construction that aims at an asymmetric graph, the on-site potentials of the sites $1$ to $4$ must all have the same value in order to maintain cospectrality of $u$ and $v$ for arbitrary $h_{1},h_{2},h_{3},h_{4}$.
As a result, the subgraph of sites $1,2,3,4$ is invariant under the cyclic permutation $1 \rightarrow 2 \rightarrow 3 \rightarrow 4 \rightarrow 1$, representing an accidental local symmetry caused by cospectrality.

Having shown a method to create graphs featuring cospectral vertices, let us now comment on how these can be modified to achieve strong cospectrality.
To this end, let us analyze \cref{eq:cospectralProjectors,eq:StronglyCospectralProjectors} which describe the conditions for cospectrality and strong cospectrality, respectively.
From these two equations, it follows that whenever a Hamiltonian $\hamil$ features two sites $u$ and $v$ which are cospectral but \emph{not} strongly cospectral, $\hamil$ must have degenerate eigenvalues.
One can thus achieve strong cospectrality of $u$ and $v$ by suitably modifying $\hamil$ such that (i) the cospectrality of $u$ and $v$ is kept and (ii) the spectrum of $\hamil$ becomes non-degenerate.
In other words, if we let $\hamil(\coParam)$ denote a Hamiltonian with cospectral sites $u$ and $v$ for a set of $N$ parameters $\coParam \subseteq \mathbb{R}^{N}$ describing couplings and on-site potentials occurring in $\hamil(\coParam)$, we look for subspaces $\coParam' \subseteq \coParam$ in which $\hamil(\coParam')$ is non-degenerate.
For the graph shown in \cref{fig:latentSymmetriesParameters} \textbf{(c)}, we have a $9$-dimensional parameter space
\begin{equation*}
	\coParam = \{(a,b,c,d,e,f,E_{\text{red}},E_{\text{blue}},E_{\text{green}}) \in \mathbb{R} \}
\end{equation*}
where $E_{\text{red}},E_{\text{blue}},E_{\text{green}}$ denote the on-site potentials of the red, blue and green sites, respectively.
For a setup designed using the procedure demonstrated in this section, the parameter space $\coParam$ can be obtained as follows.
\begin{enumerate}
	\item Parametrize the couplings and on-site potentials occurring in $ \hamil_{\text{cl}} = \hamil_{\text{cl}}(\coParam_{\text{cl}})$, where $\hamil_{\text{cl}}$ denotes the Hamiltonian describing the isolated cloud.
	For the graph depicted in \cref{fig:cospectralityDesign} \textbf{(a3)}, we have
	\begin{equation}
	\coParam_{\text{cl}} = \{(a,b,c,d,E_{\text{white}}) \in \mathbb{R} \}.
	\end{equation}
	where $E_{\text{white}}$ denotes the on-site potential of the white sites.
	\item Denote by $\coParam_{\text{coupl}}$ the parameter space for the symmetrized couplings from $\hamil_{\text{cl}}$ to the sites $u$ and $v$. For \cref{fig:cospectralityDesign} \textbf{(a3)}, we have
	\begin{equation}
	\coParam_{\text{coupl}} = \{ (e,f) \in \mathbb{R} \} .
	\end{equation}
	\item Constrain the couplings and on-site potentials occurring in $\hamil_{\overline{\text{cl}}} = \hamil_{\overline{\text{cl}}}(\coParam_{\overline{\text{cl}}})$ such that $u$ and $v$ are cospectral within $\hamil_{\overline{\text{cl}}}(\coParam_{\overline{\text{cl}}})$.
	Here $\hamil_{\overline{\text{cl}}}$ denotes the Hamiltonian describing the setup \emph{without} the cloud.
	For graphs designed using \cref{fig:cospectralityDesign}, we explicitly have
	\begin{align*}
	\coParam_{\overline{\text{cl}}}  = \begin{cases}
	\{E_r \in \mathbb{R} \} & \text{subfig. \textbf{(a3)}} \\
	\{(E_r,E_b, h_{1},h_{2}) \in \mathbb{R} \} & \text{subfig. \textbf{(b3)}} \\
	\{(E_r,E_b, h_{1}) \in \mathbb{R} \} & \text{subfig. \textbf{(c2)}} \\
	\{(E_r,E_b, h_{1},h_{2},h_{3}) \in \mathbb{R}: h_{4} > 0\} & \text{subfig. \textbf{(d)}} 
	\end{cases}
	\end{align*}
	where $E_r,E_b$ denote the on-site potentials of the red and blue sites, respectively.
	\item Construct $\coParam$ from $\coParam_{\text{cl}},\coParam_{\text{coupl}}$ and $\coParam_{\overline{\text{cl}}}$ as
	\begin{equation}
	\coParam = \coParam_{\text{cl}} \cup \coParam_{\text{coupl}} \cup \coParam_{\overline{\text{cl}}}
	\end{equation}
	so that the dimension of $\coParam$ is equal to the sum of dimensions of $\coParam_{\text{cl}},\coParam_{\text{coupl}}$ and $\coParam_{\overline{\text{cl}}}$.
	For \cref{fig:cospectralityDesign} \textbf{(a3)}, we yield
	\begin{equation*}
	\coParam = \{(a,b,c,d,e,f,E_{\text{white}},E_r) \in \mathbb{R} \} .
	\end{equation*}
\end{enumerate}
\subsection{Relating the spectral condition to minimal polynomials} \label{sec:relationMinimalPolynomials}

As explained in \cref{sec:necAndSufCond}, PGST between $u$ and $v$ happens if and only if $u$ and $v$ are strongly cospectral and the spectrum meets the conditions \cref{eq:PGSTCond1,eq:PGSTCond2,eq:PGSTCond3}.
In the previous section we showed that designing a strongly cospectral graph is straightforward. On the other hand, meeting the spectral requirements remains a difficult task.
Nevertheless, in a recent paper \cite{Eisenberg2018DMPrettyGoodQuantumState} by Eisenberg et al., this has been rendered simpler for the case of PGST.
They showed that \cref{eq:PGSTCond1,eq:PGSTCond2,eq:PGSTCond3} are \emph{automatically} fulfilled, provided that the polynomials $P_{\pm}$ (defined below) are irreducible over the base field $F$ (which contains all the entries of $\hamil$) and fulfill
\begin{equation} \label{eq:PGSTTraceCondition}
\frac{Tr(P_{+})}{deg(P_{+})} \ne \frac{Tr(P_{-})}{deg(P_{-})}
\end{equation}
where $Tr(P_{\pm})$ denote the sum of roots of $P_{\pm}$, and $deg(P_{\pm})$ denote their respective degree.
The polynomials $P_{\pm}$ stem from a decomposition of the characteristic polynomial of $\hamil$.
More specifically, given a Hamiltonian $\hamil$ with two strongly cospectral sites $u$ and $v$, its characteristic polynomial $P$ can be decomposed \cite{Eisenberg2018DMPrettyGoodQuantumState} as
\begin{equation} \label{eq:decompositionOfPolynomial}
	P = P_{0} \cdot P_{+} \cdot P_{-},
\end{equation}
such that $P_{+}$ and $P_{-}$ have no multiple roots, do not share any roots, and where the polynomials $P_{\pm}$ are related to eigenvectors of $\hamil$ which are (i) non-vanishing on sites $u$ and $v$ and (ii) are of positive/negative parity on these sites, respectively.
Each root of $P_{0}$ with multiplicity $k$ is related to exactly $k$ eigenvectors of $\hamil$, all of which have vanishing amplitudes on $u$ and $v$.
The problem of fulfilling the spectral condition for PGST thus boils down to tuning the polynomials $P_{\pm}$ accordingly.
There are two possible routes to achieve this, an indirect and a direct one.
In the indirect route, the properties of the polynomials $P_{\pm}$ are controlled by applying certain changes to the underlying Hamiltonian that cause $P_{\pm}$ to be irreducible over $F$ and meet \cref{eq:PGSTTraceCondition}, but $P_{\pm}$ are not directly known.
Such a method has been presented in Ref. \cite{Eisenberg2018DMPrettyGoodQuantumState}, where several such mechanisms have been shown.
In particular, the method shown there starts from a graph with cospectral vertices and selectively adds transcendental numbers to some diagonal entries of $\hamil$, such that the modified setup features PGST.
While elegant and powerful, indirect methods do not provide explicit forms of the polynomials $P_{\pm}$. This limits the ability to understand under which circumstances the underyling setup might feature PGST.

In cases where $\hamil$ features an involutory symmetry $\sigma$, i.e., $[H,\sigma] = 0$ with $\sigma^2 = I$, the Hamiltonian can be block-diagonalized \cite{Kempton2017QIP16210PrettyGoodQuantumState} to obtain $P_{\pm}$.
An example for such an involutory symmetry is any permutation that does only pairwise permutations of two indices, such as $\mathbb{S}: 1 \leftrightarrow 2$ (acting as the identity on indices $3$ and $4$) for the graph in \cref{fig:WalkInterpretation}.
Unfortunately, this approach is not applicable to setups that do not invoke such involutory symmetries, or where their form is unknown, such as all graphs in \cref{fig:latentSymmetriesParameters}.

In the following, we will present a new method to create PGST that relies on the recently introduced isospectral reduction of $\hamil$. This method and the transfer of compact localized states, as presented in \cref{sec:CLSTransfer}, are the two highlights of our work.
Once the polynomials are obtained, proper tuning of parameters allows to meet the requirements for PGST.

\subsubsection{Isospectral reductions} \label{sec:isospectralReductions}
\begin{figure} 
	\centering
	\includegraphics[max size={\columnwidth}{\textheight}]{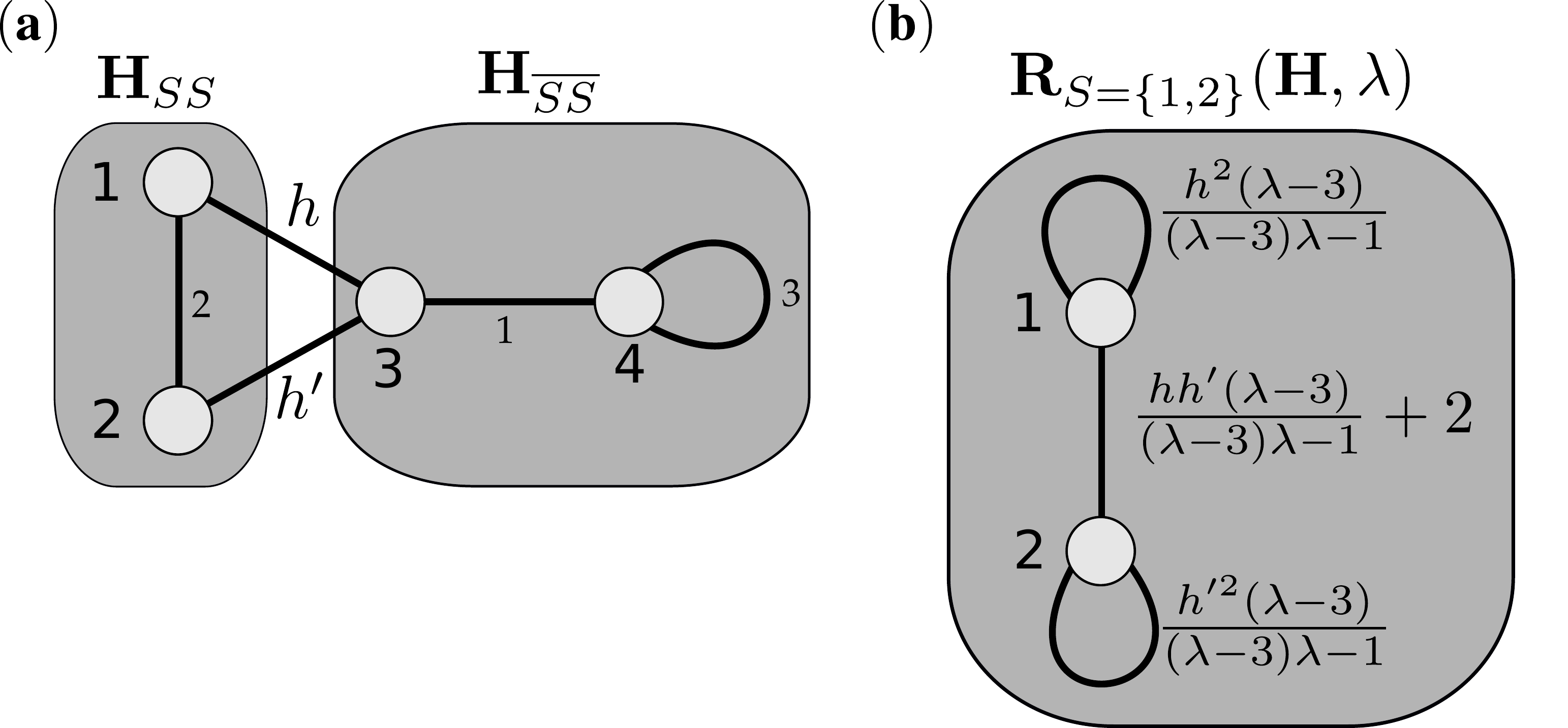}
	\caption{\textbf{(a)} Original graph, and its decomposition into $\hamil_{SS}$ and $\hamil_{\overline{SS}}$, where $S= \{1,2\}$. \textbf{(b)} The isospectral $\mathbf{R}_{S}(\hamil,\lambda)$ reduction of the graph of \textbf{(a)} over $S$.
	}
	\label{fig:isospectralReduction}
\end{figure}

We first provide some key aspects of isospectral reductions \cite{Kempton2019AMCharacterizingCospectralVerticesIsospectral,Bunimovich2011N25211IsospectralGraphTransformationsSpectral,Bunimovich2012C22033118IsospectralCompressionOtherUseful,Bunimovich2014IsospectralTransformationsNewApproach,Duarte2015LAIA474110EigenvectorsIsospectralGraphTransformations,VasquezFernandoGuevara2014NLAwA22145PseudospectraIsospectrallyReducedMatrices}, introduced first by Bunimovich and Webb \cite{Bunimovich2011N25211IsospectralGraphTransformationsSpectral}.
This concept will allow us to extract the polynomials $P_{\pm}$. Our explanations will be accompanied by the illustration in \cref{fig:isospectralReduction}.

The basic idea of isospectral reductions is to reduce the dimension of a given matrix Hamiltonian $\hamil$ by certain transformations specified by a set of sites $S$, yielding a smaller matrix $\mathcal{\mathbf{R}}_{S}(\hamil,\lambda)$ dependent on a parameter $\lambda$, which carries the same or almost the same spectral information as the original matrix $\hamil$. Among others, the benefit of such a reduction lies in a reduction of complexity. For this reason, the isospectral reduction has been invented in the context of network analysis, where the sheer size of the investigated networks often complicates their treatment.
Let us now define the isospectral reduction of a given matrix $\hamil$.
This reduction $\mathcal{\mathbf{R}}_{S}(H,\lambda)$ is done over the set of sites $S$, so that
\begin{equation} \label{eq:isospectralReduction}
\mathcal{\mathbf{R}}_{S}(\hamil,\lambda) = \hamil_{SS} - \hamil_{S\overline{S}}\left(\hamil_{\overline{S}\overline{S}} - \mathbf{I} \lambda \right)^{-1} \hamil_{\overline{S}S}
\end{equation}
and is defined for all values of $\lambda$ that are not eigenvalues of $\hamil_{\overline{S}\overline{S}}$, where $\overline{S}$ denotes the complement of the set of vertices $S$.
$\hamil_{SS}$ and $\hamil_{\overline{SS}}$ denote two subsystems of $\hamil$, obtained from $\hamil$ by deleting all sites in $\overline{S}$ or $S$, respectively. $\hamil_{S\overline{S}} = (\hamil_{\overline{S}S})^{T}$ are the submatrices which couple $\hamil_{SS}$ to $\hamil_{\overline{SS}}$ and $\hamil_{\overline{SS}}$ to $\hamil_{SS}$, respectively.
The dimension of $\mathcal{\mathbf{R}}_{S}(\hamil,\lambda)$ is given by $|S|$, i.e., the number of sites over which $\hamil$ is isospectrally reduced.
Such a decomposition is shown in \cref{fig:isospectralReduction} \textbf{(a)}.
In \cref{fig:isospectralReduction} \textbf{(b)}, we then show the isospectral reduction of the graph in \cref{fig:isospectralReduction} \textbf{(a)} over the sites $S = \{1,2\}$.

A major goal of the isospectral reduction is to reduce the size of the problem, whilst maintaining (almost) all of its spectral features. It may seem that such a reduction is impossible, since, by the fundamental theorem of algebra, a hermitian matrix $\hamil \in \mathbb{C}^{N \times{}N}$ has exactly $N$ eigenvalues. A reduced version $\hamil' \in \mathbb{C}^{|S| \times{} |S|},\; |S| < N$ would, therefore, inevitably have $N - |S|$ less eigenvalues.
However, the above is not necessarily true anymore if the entries of $\hamil'$ are not just constant real or complex numbers, but rational functions of a parameter $\lambda$.
This is the case for $\mathcal{\mathbf{R}}_{S}(\hamil,\lambda)$, as can be seen for example in \cref{fig:isospectralReduction} \textbf{(b)}.
This change in the nature of matrix entries also slightly alters the definition of eigenvalues of the isospectrally reduced matrix $\mathcal{\mathbf{R}}_{S}(\hamil,\lambda)$ compared to that of matrices with constant entries.
While the eigenvalues $\lambda_{i}$ of such a matrix $\hamil$ fulfill
\begin{equation*}
	det\Big(\hamil - \mathbf{I} \lambda_{i}  \Big) = 0,
\end{equation*}
the eigenvalues $\lambda_{i}$ of $\mathcal{\mathbf{R}}_{S}(\hamil,\lambda)$ fulfill
\begin{equation} \label{eq:isospectralEigenvalues}
det\Big(\mathcal{\mathbf{R}}_{S}(\hamil,\lambda_{i}) - \mathbf{I} \lambda_{i} \Big) = 0 .
\end{equation}
It can then be shown that the set of eigenvalues of $\mathcal{\mathbf{R}}_{S}(\hamil,\lambda)$ contains all eigenvalues of $\hamil$, except those which are also eigenvalues of $\hamil_{\overline{S}\overline{S}}$.
Thus, if $\hamil$ and $\hamil_{\overline{S}\overline{S}}$ do not share any eigenvalues, the spectrum $\sigma(\mathcal{\mathbf{R}}_{S}(\hamil,\lambda))$ is identical with that of $\hamil$, i.e., $\sigma(\mathcal{\mathbf{R}}_{S}(\hamil,\lambda)) = \sigma(\hamil)$, as desired.

Similar to the definition of eigenvalues of $\mathcal{\mathbf{R}}_{S}(\hamil,\lambda)$, as done in \cref{eq:isospectralEigenvalues}, it is also possible to generalize the concept of \emph{eigenvectors} to isospectral reductions $\mathcal{\mathbf{R}}_{S}(\hamil,\lambda)$.
These eigenvectors $\{\mathbf{v}_{1},\ldots{},\mathbf{v}_{n} \}$ of $\mathcal{\mathbf{R}}_{S}(\hamil,\lambda)$, where $n$ is the number of eigenvalues of $\mathcal{\mathbf{R}}_{S}(\hamil,\lambda)$, fulfill
\begin{equation*}
(\mathcal{\mathbf{R}}_{S}(\hamil,\lambda_{i}) - \mathbf{I} \lambda_{i}) \mathbf{v}_{i} = 0 .
\end{equation*}
Contrary to eigenvectors of the symmetric matrix $\hamil$, the set $\{\mathbf{v}_{i} \}$ does not need to be pairwise orthogonal, and could even be linearly dependent or pairwise identical.
Their importance stems from the fact that they can be linked \cite{Bunimovich2018LAIA551104GeneralizedEigenvectorsIsospectralTransformations,Duarte2015LAIA474110EigenvectorsIsospectralGraphTransformations} to the eigenvectors of $\hamil$.
Namely, every eigenvector $\mathbf{v}_{i} \in \mathbb{R}^{|S|\times{} 1}$ of $\mathcal{\mathbf{R}}_{S}(\hamil,\lambda)$ with eigenvalue $\lambda_{i}$ is, up to normalization, the \emph{projection} of the corresponding eigenvector $\mathbf{V}_{i} \in \mathbb{R}^{N \times{} 1}$ of $\hamil \in \mathbb{R}^{N \times{} N}$ onto the sites $S$, i.e., equal to $\big(\mathbf{V}_{i}\big)_{S} \in \mathbb{R}^{|S|\times 1}$, where $\hamil \mathbf{V}_{i} = \lambda_{i} \mathbf{V}_{i}$ and $|S|$ denotes the number of elements in $S$.

\subsubsection{Extracting the polynomials $P_{\pm}$ through isospectral reductions} \label{sec:isospectralExtraction}

In order to use the isospectral reduction to extract the polynomials $P_{\pm}$, let us now investigate the special case of an isospectral reduction over two sites in more detail.
Such a reduction is shown in \cref{fig:isospectralReduction} \textbf{(b)} for the setup shown in \cref{fig:isospectralReduction} \textbf{(a)}.
Inspecting this isospectral reduction $\mathcal{\mathbf{R}}_{S}(\hamil,\lambda)$ in more detail, we see that the respective ``on-site potentials'' $\frac{h^2 (\lambda -3)}{(\lambda -3) \lambda -1}$ and $\frac{h'^2 (\lambda -3)}{(\lambda -3) \lambda -1}$ of sites $1$ and $2$ become equal for all $\lambda$ if and only if $h = \pm h'$. In this case
\begin{equation} \label{eq:bisymmetricExample}
	\mathcal{\mathbf{R}}_{S}(\hamil,\lambda) = \left(
	\begin{array}{cc}
	\frac{h^2 (\lambda -3)}{(\lambda -3) \lambda -1} & 2 \pm \frac{(\lambda -3)
		h^2}{(\lambda -3) \lambda -1} \\
2 \pm	\frac{(\lambda -3) h^2}{(\lambda -3) \lambda -1} & \frac{h^2 (\lambda
		-3)}{(\lambda -3) \lambda -1} \\
	\end{array}
	\right)
\end{equation}
becomes bisymmetric, i.e., symmetric about both the diagonal and the anti-diagonal.
Interestingly, the choice $h' = \pm h$ is also the only one that makes $u$ and $v$ cospectral.
This finding suggests that there might be a connection between the bisymmetry of $\mathcal{\mathbf{R}}_{\{u,v\}}(\hamil,\lambda)$ and cospectrality of $u$ and $v$.
This is indeed the case, as was very recently proven in Ref. \cite{Kempton2019AMCharacterizingCospectralVerticesIsospectral}. For symmetric matrices $\hamil$, the isospectral reduction $\mathcal{\mathbf{R}}_{S}(\hamil,\lambda)$ over two sites $\{u,v\} = S$ is bisymmetric if and only if $u$ and $v$ are cospectral in $\hamil$.
Moreover, $u$ and $v$ are strongly cospectral if and only if they are cospectral and all eigenvalues of $\mathcal{\mathbf{R}}_{S}(\hamil,\lambda)$ are simple.
This theorem is remarkable, as it connects the two seemingly unrelated concepts of cospectrality and isospectral reductions.

To give an intuitive argument for why this theorem makes sense, we show how cospectrality of $u$ and $v$ follows from bisymmetry of $\mathcal{\mathbf{R}}_{S = \{u,v\}}(\hamil,\lambda)$ for the simple case when $\sigma(\hamil) = \sigma(\mathcal{\mathbf{R}}_{S}(\hamil,\lambda))$.
In this case, each eigenvector $\mathbf{v}_{i}$ of $\mathcal{\mathbf{R}}_{S}(\hamil,\lambda)$ is the projection of the corresponding eigenvector $\mathbf{V}_{i}$ of $\hamil$ on the sites $S$.
Now, as can be easily shown, the eigenvectors of $\mathcal{\mathbf{R}}_{S}(\hamil,\lambda)$ have (in the case of degeneracies, can be chosen to have) parity $\pm 1$ on $u$ and $v$ if and only if $\mathcal{\mathbf{R}}_{S}(\hamil,\lambda)$ is bisymmetric.
Therefore, the eigenvectors $\{\mathbf{V}_{i}\}$ of $\hamil$ fulfill \cref{eq:cospectralProjectors}, i.e., sites $u,v$ are cospectral due to the bisymmetry of $\mathcal{\mathbf{R}}_{S}(\hamil,\lambda)$. If, additionally, all eigenvalues of $\mathcal{\mathbf{R}}_{S}(\hamil,\lambda)$ are simple, the $\{\mathbf{V}_{i}\}$ also fulfill \cref{eq:StronglyCospectralProjectors}, i.e., $(\mathbf{V}_{i})_{u} = \pm (\mathbf{V}_{i})_{v}$.

We now use the connection between cospectrality and bisymmetry of $\mathcal{\mathbf{R}}_{S = \{u,v\}}(\hamil,\lambda)$ to extract $P_{\pm}$.
To this end, we assume that $u$ and $v$ are strongly cospectral. By theorem 3.8. from Ref. \cite{Kempton2019AMCharacterizingCospectralVerticesIsospectral}, $\mathcal{\mathbf{R}}_{S}(\hamil,\lambda)$ is then bisymmetric, and all its eigenvalues are simple. Due to its bisymmetry, we can parametrize
\begin{equation}
	\mathcal{\mathbf{R}}_{S}(\hamil,\lambda) = \begin{pmatrix}
	A(\lambda) & B(\lambda) \\
	B(\lambda) & A(\lambda)
	\end{pmatrix}
\end{equation}
with $A(\lambda),B(\lambda)$ being rational functions of $\lambda$.
As we have explained above, all eigenvectors of $\mathcal{\mathbf{R}}_{S}(\hamil,\lambda)$ are (in the case of degeneracies, can be chosen to be) of definite parity on $u$ and $v$.
Therefore, the characteristic polynomial $P_{\mathcal{\mathbf{R}}}(\lambda)$ of $\mathcal{\mathbf{R}}_{S}(\hamil,\lambda)$ can be factored into two parts, $P_{\mathcal{\mathbf{R}}} = P_{\mathcal{\mathbf{R}}}^{+}(\lambda) \cdot P_{\mathcal{\mathbf{R}}}^{-}(\lambda)$, such that the roots of the polynomials $P_{\pm}(\lambda)$ are the eigenvalues of eigenvectors of $\mathbf{R}_{S}(\hamil,\lambda)$ with positive and negative parity, respectively.
As we show in \cref{app:ProofOfPPrime}, $P_{\mathcal{\mathbf{R}}}^{\pm} = A(\lambda) \pm B(\lambda) - \lambda$.
They obey the relation
\begin{equation*}
P_{\mathcal{\mathbf{R}}}^{+}(\lambda) \cdot P_{\mathcal{\mathbf{R}}}^{-}(\lambda) = det\Big(\mathcal{\mathbf{R}}_{S}(\hamil,\lambda) - \mathbf{I} \lambda \Big) =  \frac{det(\hamil - \mathbf{I} \lambda)}{det(\hamil_{\overline{SS}} - \mathbf{I} \lambda)}
\end{equation*}
where the first equality is proven in \cref{app:ProofOfPPrime}, and the second on p. $7$ in Ref. \cite{Bunimovich2014IsospectralTransformationsNewApproach}.

There are now two possible scenarios for which the polynomials $P_{\pm}(\lambda)$ can be obtained.
In the first scenario, $\hamil$ and $\hamil_{\overline{SS}}$ must not share any eigenvalues.
In that case, all eigenvalues of $\hamil$ are given by the union of roots of $P_{\mathcal{\mathbf{R}}}^{\pm}(\lambda)$, and by the above assumption of strong cospectrality of $u$ and $v$, all these eigenvalues are non-degenerate.
Combining these properties, we see that \emph{all} eigenvectors of $\hamil$ do not vanish on the sites $u$ and $v$, and the corresponding amplitudes on these two sites are of definite parity w.r.t. exchanging $u$ and $v$.
Thus, $P_{0}(\lambda) = 1$ [from the decomposition of the characteristic polynomial of $\hamil$, as done in \cref{eq:decompositionOfPolynomial}], and we get
\begin{equation*}
	P(\lambda) = P_{+}(\lambda) \cdot P_{-}(\lambda)
\end{equation*}
where $P(\lambda) = det(\hamil - \mathbf{I} \lambda)$ is the characteristic polynomial of $\hamil$.
The $P_{\mathcal{\mathbf{R}}}^{\pm} \equiv p_{\pm}/q_{\pm}$ are rational functions in $\lambda$, so that
\begin{equation} \label{eq:polynomialProduct}
P_{\mathcal{\mathbf{R}}}^{+}(\lambda) \cdot P_{\mathcal{\mathbf{R}}}^{-}(\lambda) =\frac{p_{+}(\lambda)}{q_{+}(\lambda)} \cdot \frac{p_{-}(\lambda)}{q_{-}(\lambda)} = \frac{det(\hamil - \mathbf{I} \lambda)}{det(\hamil_{\overline{SS}} - \mathbf{I} \lambda)}
\end{equation}
where $p_{\pm}(\lambda),q_{\pm}(\lambda)$ and both determinants are polynomials in $\lambda$.
Since
\begin{equation*}
	det(\hamil - \mathbf{I} \lambda) = P(\lambda) = P_{+}(\lambda) \cdot P_{-}(\lambda)
\end{equation*}
it would be ideal if the numerators in \cref{eq:polynomialProduct} match, so that $p_{\pm}(\lambda) = P_{\pm}(\lambda)$.
However, since \cref{eq:polynomialProduct} remains invariant under the transformation
\begin{align} \label{eq:reduceFraction1}
	p_{\pm}(\lambda) &\rightarrow c_{\pm}(\lambda) \cdot p_{\pm}(\lambda) \\
	q_{\pm}(\lambda) &\rightarrow c_{\pm}(\lambda) \cdot q_{\pm}(\lambda) \label{eq:reduceFraction2}
\end{align}
with $c_{\pm}(\lambda)$ functions of $\lambda$, the $p_{\pm}(\lambda)$ are not uniquely determined by \cref{eq:polynomialProduct} alone.
To uniquely determine $P_{\pm}(\lambda)$, one needs to properly reduce the fractions $p_{\pm}(\lambda)/q_{\pm}(\lambda)$ [i.e., performing the transformations of \cref{eq:reduceFraction1,eq:reduceFraction2} with suitable $c_{\pm}(\lambda)$] such that the following conditions are fulfilled.
Firstly, the leading-order coefficients $a_{n_{\pm}}^{(\pm)}$ of the polynomials
\begin{align*}
	p_{+}(\lambda) &= \sum_{n=0}^{n_{+}} a_{n}^{(+)} \lambda^{n} \\
	p_{-}(\lambda) &= \sum_{n=0}^{n_{-}} a_{n}^{(-)} \lambda^{n}
\end{align*}
where $n_{\pm}$ are the respective degrees of $p_{\pm}(\lambda)$, must be chosen such that
\begin{align*}
	a_{n_{+}}^{(+)} &= 1 \\
	a_{n_{-}}^{(-)} &= (-1)^{N}
\end{align*}
where $N$ is the dimension of $\hamil \in \mathbb{R}^{N \times N}$.
This ensures that the product of $p_{+}(\lambda) \cdot p_{-}(\lambda)$ has a leading order coefficient of $(-1)^{N}$, which matches the leading order coefficient of $det(\hamil - \mathbf{I} \lambda)$.
Secondly, the fractions $p_{\pm}(\lambda)/q_{\pm}(\lambda)$ must be \emph{irreducible}.
The latter property means that $p_{+}(\lambda),q_{+}(\lambda)$ [and also $p_{-}(\lambda),q_{-}(\lambda)$] are coprime, i.e., their only common factor is unity.
If the above two conditions are fulfilled, we obtain
\begin{equation*} 
	P_{\pm}(\lambda) = p_{\pm}(\lambda)
\end{equation*}
as desired.

The second scenario where $P_{\pm}(\lambda)$ can be obtained is when all eigenvectors $\mathbf{x'}_{i}$ of $\hamil$ which are related to common eigenvalues $\lambda'_{i}$ of both $\hamil$ and $\hamil_{\overline{SS}}$ vanish on $S$.
The polynomial $P_{0}(\lambda)$ from \cref{eq:decompositionOfPolynomial} then becomes
\begin{equation} \label{eq:PZero}
P_{0}(\lambda) = \prod_{i} (\lambda'_{i} -  \lambda)
\end{equation}
and we can factorize
\begin{align}
	det(\hamil - \mathbf{I} \lambda) &= \Big(\prod_{i} ( \lambda_{i} - \lambda)\Big) \cdot \Big(\prod_{i} ( \lambda'_{i} - \lambda)\Big) \label{eq:DecompositionOfHEigenvalues}, \\
	det(\hamil_{\overline{SS}} - \mathbf{I} \lambda) &= \Big(\prod_{i} ( \lambda'_{i} - \lambda)\Big) \cdot \Big(\prod_{i} ( \lambda''_{i} - \lambda)\Big) \nonumber
\end{align}
where $\lambda''_{i}$ are the eigenvalues of $\hamil_{\overline{SS}}$ which are not simultaneously eigenvalues of $\hamil$.
As a result of \cref{eq:PZero,eq:decompositionOfPolynomial,eq:DecompositionOfHEigenvalues},
\begin{equation*}
	\prod_{i} ( \lambda_{i} - \lambda) = P_{+}(\lambda) \cdot P_{-}(\lambda)
\end{equation*}
and similarly to \cref{eq:polynomialProduct}, we obtain
\begin{equation*}
P_{\mathcal{\mathbf{R}}}^{+}(\lambda) \cdot P_{\mathcal{\mathbf{R}}}^{-}(\lambda)= \frac{p_{+}(\lambda)}{q_{+}(\lambda)} \cdot \frac{p_{-}(\lambda)}{q_{-}(\lambda)} = \frac{\prod_{i}(\lambda_{i} - \lambda ) }{\prod_{i}( \lambda''_{i} - \lambda) }.
\end{equation*}
If the fractions $p_{\pm}(\lambda)/q_{\pm}(\lambda)$ are properly reduced as above, we again have that
\begin{equation*}
P_{\pm}(\lambda) = p_{\pm}(\lambda).
\end{equation*}

The isospectral reduction can thus be used to extract the polynomials $P_{\pm}$ provided that (i) $\hamil$ and $\hamil_{\overline{SS}}$ do not share a common root, or (ii) all common roots of $\hamil$ and $\hamil_{\overline{SS}}$ are related to eigenvectors of $\hamil$ which vanish on $S$.
In the next section, we show how this knowledge can be harnessed to design Hamiltonians featuring PGST.

\section{Application: Designing graphs with pretty good state transfer} \label{sec:Application}
In the previous \cref{sec:MathematicalFoundations}, we have introduced the concept of cospectrality and have shown how, based on the isospectral reduction, the polynomials $P_{\pm}$ can be extracted.
With this theoretical background, one can derive the following algorithm for the design of graphs with PGST.
\begin{enumerate}
	\item \textbf{Achieving cospectrality}\\Design/take a graph $\hamil(\coParam)$ with cospectral vertices $u$ and $v$, e.g., by means of the procedure demonstrated in \cref{sec:designingCospectral}.
	Here, $\coParam$ denotes the parameter space of couplings and on-site potentials occurring in $\hamil$ for which $u$ and $v$ are cospectral and for which there exists at least one possible walk from $u$ to $v$.
	For example, for the graph depicted in \cref{fig:latentSymmetriesParameters} \textbf{(b)}, we have $\coParam = \{(a,b,c,d,E_{\text{red}},E_{\text{blue}}) \in \mathbb{R} \; : \; b \ne 0 \; \text{or}\;  a c \ne 0 \}$, where $E_{\text{red}},E_{\text{blue}}$ denote the on-site potentials of the red and blue sites, respectively.
	\item \textbf{Achieving strong cospectrality}\\Due to cospectrality of $u$ and $v$ for all $\hamil(\coParam)$, the isospectral reduction
	\begin{equation*}
		\mathcal{\mathbf{R}}_{S=\{u,v\}}(\hamil(\coParam),\lambda) = \begin{pmatrix}
		A(\coParam, \lambda) & B(\coParam, \lambda) \\
		B(\coParam, \lambda) & A(\coParam, \lambda)
		\end{pmatrix}
	\end{equation*}
	 of $\hamil(\coParam)$ over $S=\{u,v\}$ [with rational functions $A(\coParam,\lambda),B(\coParam,\lambda)$] is, by Theorem 3.3. from Ref. \cite{Kempton2019AMCharacterizingCospectralVerticesIsospectral}, guaranteed to be bisymmetric.
	 Compute
	 \begin{equation*}
	 	P_{\mathbf{R}}^{\pm}(\coParam,\lambda) = A(\coParam,\lambda) \pm B(\coParam,\lambda) - \lambda \equiv \frac{p_{\pm}(\coParam,\lambda)}{q_{\pm}(\coParam,\lambda)},
	 \end{equation*}
and, by suitable algorithms [see the next \cref{sec:AlgorithmComments}], find a subspace $\coParam' \subseteq \coParam$ for which $P_{\mathbf{R}}^{\pm}(\coParam',\lambda)$ individually have only simple roots, and additionally have no common roots.
By Theorem 3.8 from Ref. \cite{Kempton2019AMCharacterizingCospectralVerticesIsospectral}, $u$ and $v$ are then strongly cospectral.

\item \textbf{Extraction of $P_{\pm}$}\\
By suitable algorithms [see the next \cref{sec:AlgorithmComments}], either
\begin{itemize}
	\item further restrict $\coParam'$ such that $\hamil(\coParam')$ and $\hamil_{\overline{SS}}(\coParam')$ do not share any eigenvalues,
	\item or, alternatively, restrict $\coParam'$ such that all eigenvalues $\lambda'_{i}$ shared by $\hamil(\coParam')$ and $\hamil_{\overline{SS}}(\coParam')$ are related to eigenvectors of $\hamil(\coParam')$ which \emph{vanish on $S$}.
\end{itemize}
In both cases, properly reduce (or expand) the fractions occurring in $P_{\mathbf{R}}^{\pm}(\coParam',\lambda)$ such that
\begin{itemize}
	\item the leading order coefficients of $p_{\pm}(\coParam',\lambda)$ [which are polynomials in $\lambda$] are $+1$ and $(-1)^{N}$, respectively, where $N$ is the dimension of $\hamil \in \mathbb{R}^{N \times N}$,
	\item $p_{\pm}(\coParam',\lambda)/q_{\pm}(\coParam',\lambda)$ are irreducible.
\end{itemize}
As a result $P_{\pm}(\coParam',\lambda) = p_{\pm}(\coParam',\lambda)$.
	\item \textbf{Enforcing pretty good state transfer}\\
	Within the subspace $\coParam'$, search [see the next \cref{sec:AlgorithmComments}] for realizations $\coParam'' \subseteq \coParam'$ such that
	\begin{enumerate}
		\item $P_{\pm}(\coParam'',\lambda)$ are irreducible over the base field $F$ which contains all entries of $\hamil(\coParam'')$.
		\item $\frac{Tr(P_{+}(\coParam'',\lambda))}{deg(P_{+}(\coParam'',\lambda))} \ne \frac{Tr(P_{-}(\coParam'',\lambda))}{deg(P_{-}(\coParam'',\lambda))}$.
	\end{enumerate}
	$\hamil(\coParam'')$ then features PGST from $u$ to $v$.
	We note that $Tr(P_{\pm}(\coParam'',\lambda))$ can be computed without finding the roots of these polynomials, since $Tr(f(x)) = -a_{n-1}/a_{n}$ for a polynomial $f(x) = \sum_{i=0}^{n} a_{i} x^{i}$ of degree $n$.
	\item \textbf{Repetition (if necessary)}\\
	Since not every graph may support PGST, the above procedure is not guaranteed to work in all cases [see the next \cref{sec:AlgorithmComments} for details].
	Thus, if step 4. is not successful, i.e., no parameters $\coParam''$ exist such that $P_{\pm}(\coParam'',\lambda)$ fulfill 4. (a) and (b), go back to step 3. and try its alternative route. If this, again, is not successful, go back to step 1., modify the graph by adding/removing vertices and start anew.
\end{enumerate}

\subsection{Annotations} \label{sec:AlgorithmComments}
Let us now make two comments regarding the above algorithm.
Firstly, the steps 2. to 4. require the search for suitable subspaces, which in general must be performed by means of suitable trial-and-error algorithms.
However, the subspace $\coParam' \subseteq \coParam$ [the search for which is the subject of steps 2. and 3. of the algorithm] can in some cases be given by explicit expressions, as we demonstrate in the next \cref{sec:Example}.
Secondly, not all setups may support PGST, and the above algorithm is therefore not guaranteed to work in all cases.
However, we have successfully tested the algorithm with a variety of setups, and among others, all six graphs depicted in \cref{fig:latentSymmetriesParameters} were successfully tuned to support PGST between the two red sites.

Overall, we stress that the main advantage of our algorithm, compared to existing methods for the design of PGST, is the ability to derive \emph{explicit} forms for the polynomials $P_{\pm}$.
We hope that the insights gained on how to extract the polynomials $P_{\pm}$ will lead to a better understanding on the classes of setups which support PGST.
This understanding is facilitated by the fact that the core method of our approach, the isospectral reduction $\mathcal{\mathbf{R}}_{S=\{u,v\}}(\hamil(\coParam),\lambda)$, can be performed symbolically. As we will see in the next section, in some cases, nearly all steps of the algorithm can be done without numerical evaluations at all.

\subsection{Example} \label{sec:Example}
We now apply the algorithm presented above to a simple example, and will go separately through each of the steps 1. to 4.
\begin{figure} 
	\centering
	\includegraphics[max size={0.6\columnwidth}{\textheight}]{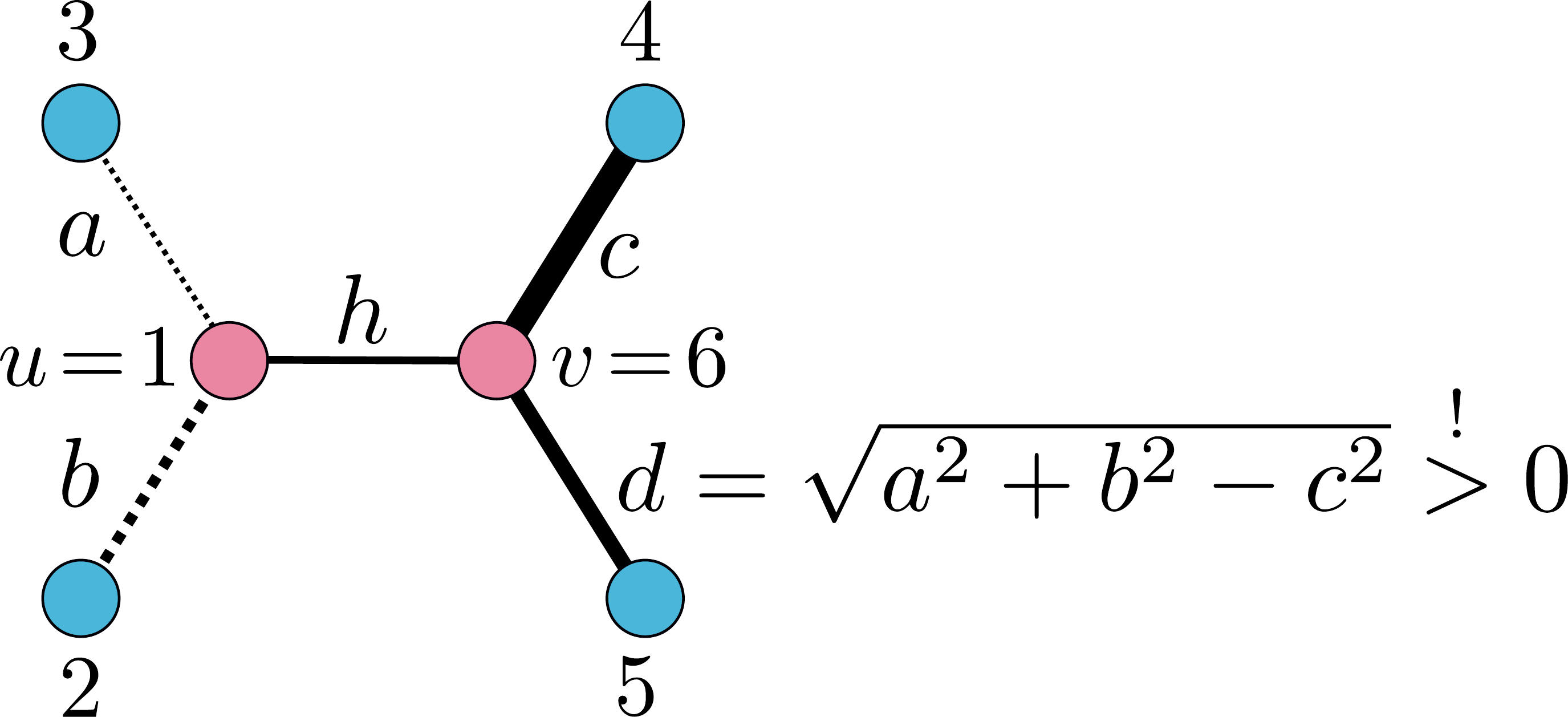}
	\caption{A simple system which can be tuned to feature PGST between sites $u=1$ and $v=6$, as demonstrated in \cref{sec:Example}.
	}
	\label{fig:example1}
\end{figure}

\subsubsection{Achieving cospectrality}
We start the algorithm with the graph shown in \cref{fig:example1}, which represents a very simple graph featuring cospectral vertices $u$ and $v$.
The graph is described by
\begin{equation} \label{eq:example1Hamiltonian}
	\hamil(\coParam) = \left(
	\begin{array}{cccccc}
	E_r & a & b & 0 & 0 & h \\
	a & E_b & 0 & 0 & 0 & 0 \\
	b & 0 & E_b & 0 & 0 & 0 \\
	0 & 0 & 0 & E_b & 0 & c \\
	0 & 0 & 0 & 0 & E_b & d \\
	h & 0 & 0 & c & d & E_r \\
	\end{array}
	\right)
\end{equation}
where $d = \sqrt{a^2 +b^2 -c^2}$,
\begin{equation*}
	\coParam = \{(a,b,c,h,E_{r},E_{b}) \in \mathbb{R}\; :\; d > 0 \; \text{and} (a,b,c,h) \ne 0\},
\end{equation*}
and $E_{b},E_{r}$ denote the on-site potentials of the blue and red sites, respectively.
The sites $u=1$ and $v=6$ are then guaranteed to be cospectral for the Hamiltonian $\hamil(\coParam)$.

The form of $d = \sqrt{a^2 +b^2 -c^2}$ is chosen such as to ensure cospectrality of $u$ and $v$ within a large parameter space. For example, for $a=b=c$, $\hamil$ would be invariant under the exchange $3 \leftrightarrow 4, 2 \leftrightarrow 5, 1 \leftrightarrow 6$, so that $u=1$ and $v=6$ would trivially be cospectral.
However, our choice of $d = \sqrt{a^2 +b^2 -c^2} > 0$ ensures this cospectrality even for asymmetric cases such as $a = 2b = 4c$, where $\hamil$ is not invariant under any non-trivial permutation of sites.

\subsubsection{Achieving strong cospectrality}
The isospectral reduction of $\hamil(\coParam)$ over $S = \{u,v\}$ then gives
\begin{equation*}
	\mathcal{\mathbf{R}}_{S=\{u,v\}}(\hamil(\coParam),\lambda) = \left(
	\begin{array}{cc}
	\frac{\delta }{\lambda -E_b}+E_r & h \\
	h & \frac{\delta }{\lambda -E_b}+E_r \\
	\end{array}
	\right)
\end{equation*}
where $\delta = a^2 + b^2$,
so that
\begin{align*}
	P_{\mathbf{R}}^{+}(\coParam,\lambda) =& \frac{\delta }{\lambda -E_b}+h-\lambda +E_r , \\
	P_{\mathbf{R}}^{-}(\coParam,\lambda) =& \frac{\delta }{\lambda -E_b}-h-\lambda +E_r .
\end{align*}

Following the procedure of the algorithm, we now have to investigate (i) under which circumstances all roots of $P_{\mathbf{R}}^{\pm}(\coParam,\lambda)$ are simple and (ii) under which conditions $P_{\mathbf{R}}^{+}(\coParam,\lambda)$ and $P_{\mathbf{R}}^{-}(\coParam,\lambda)$ do not share any roots. Since
\begin{align} \label{eq:example1Pplus}
	P_{\mathbf{R}}^{+}(\coParam,\lambda) &= \frac{-E_b \left(h-\lambda +E_r\right)+\lambda  (h-\lambda )+\delta +\lambda  E_r}{\lambda -E_b}, \\
	P_{\mathbf{R}}^{-}(\coParam,\lambda) &= \frac{E_b \left(h+\lambda -E_r\right)-\lambda  (h+\lambda )+\delta +\lambda 
		E_r}{\lambda -E_b} \label{eq:example1Pminus}
\end{align}
are rational functions in $\lambda$, we define the corresponding numerators and denominators as $p_{\pm}(\coParam,\lambda)$ and $q_{\pm}(\coParam,\lambda)$. Since the $p_{+}(\coParam,\lambda), q_{+}(\coParam,\lambda)$ and $p_{-}(\coParam,\lambda), q_{-}(\coParam,\lambda)$ could in principle share roots, we need to evaluate when this can happen.
To this end, we can use the so-called \emph{resultant} \cite{WeissteinFMWWRResultantMathWorldAWolframWeb}.
Two given polynomials $f(x)$ and $g(x)$ share at least one root if and only if their resultant $R(f,g)$ is zero. The resultant, defined in terms of the so-called Sylvester-Matrix, can be computed symbolically and is implemented in common computer algebra systems.
For the problem at hand, we yield
\begin{equation}
	R(p_{+}(\coParam,\lambda),q_{+}(\coParam,\lambda)) = R(p_{-}(\coParam,\lambda),q_{-}(\coParam,\lambda)) = \delta
\end{equation}
which can obviously never vanish, since $\delta = a^2 + b^2$ and we demanded that $a,b \in \mathbb{R}$ and $a,b \ne 0$. Thus, we can evaluate the roots of $P_{\mathbf{R}}^{\pm}(\coParam,\lambda)$ by evaluating only their numerators $p_{\pm}(\coParam,\lambda)$.

To check whether $p_{+}(\coParam,\lambda)$ and $p_{-}(\coParam,\lambda)$ share any roots, we again rely on the resultant, which gives
\begin{align*}
R(p{+}(\coParam,\lambda),p_{-}(\coParam,\lambda))&= 4 h^2 \delta > 0.
\end{align*}
Thus, $p_{+}(\coParam,\lambda)$ and $p_{-}(\coParam,\lambda)$ will not share any roots.
We then need to check when $p_{+}(\coParam,\lambda)$ and $p_{-}(\coParam,\lambda)$ individually have multiple roots.
To this end, we compute their so-called \emph{discrimant} \cite{WeissteinPolynomialDiscriminantMathWorldAWolfram}. The discrimant $D\big(f(x)\big)$ of a polynomial $f(x)$ is zero if and only if $f(x)$ has at least one multiple root. Like the resultant, the discriminant can be computed analytically and is implemented in many computer algebra systems.
We then get
\begin{equation*}
D\big(p_{\pm}(\coParam,\lambda) \big) = \left(E_b-E_r\right) \left(E_b\mp 2 h-E_r\right)+h^2+4 \delta .
\end{equation*}
$D\big(p_{\pm}(\coParam,\lambda) \big)$ can only vanish if $\delta =-\frac{1}{4} \left(h \mp v_{v} \pm v_{r}\right){}^2 < 0$, which is again forbidden by our assumptions that $a,b \in \mathbb{R}$ and $a,b \ne 0$.

Let us now recapitulate the above. We have investigated under which conditions all roots of $\mathcal{\mathbf{R}}_{S=\{u,v\}}(\hamil(\coParam),\lambda)$ are simple.
The motivation for this study is the fact that, whenever this is the case, the sites $u$ and $v$ are not only cospectral, but also \emph{strongly cospectral}.
For the Hamiltonian given by \cref{eq:example1Hamiltonian}, we have found that both of the above conditions are fulfilled for all elements in the parameter space $\coParam$, so that $\coParam' = \coParam$, and $u,v$ are always strongly cospectral in this Hamiltonian $\hamil(\coParam)$.
We can thus move on to the third step of our algorithm.

\subsubsection{Extraction of $P_{\pm}$}

Following the procedure of the algorithm, we now have to investigate under which circumstances $\hamil(\coParam')$ and $\hamil_{\overline{SS}}(\coParam')$ share eigenvalues. We therefore compute their resultant
\begin{equation}
	R\big( det(\hamil(\coParam') - \mathbf{I} \lambda), det(\hamil_{\overline{SS}}(\coParam') - \mathbf{I} \lambda) \big) = 0.
\end{equation}
Thus, $\hamil(\coParam')$ and $\hamil_{\overline{SS}}(\coParam')$ always share at least one eigenvalue.
Indeed, closer evaluation shows that, irrespective of how $\coParam'$ is chosen, $\hamil(\coParam')$ and $\hamil_{\overline{SS}}(\coParam')$ share a twofold degenerate eigenvalue $\lambda = E_{b}$, with corresponding (unnormalized) eigenvectors $\mathbf{x}_{1} = (0,1,-b/a,0,0,0)^{T}/\sqrt{4}$ and $\mathbf{x}_{2} = (0,0,0,1,-c/d,0)^{T}/\sqrt{4}$. Both eigenvectors have zero amplitude on the sites $S = \{1,6\}$, and by the reasoning in \cref{sec:isospectralExtraction}, the corresponding doubly degenerate eigenvalue $E_{b}$ of $\hamil(\coParam')$ is not of relevance to us.
To see whether there are any other common eigenvalues of $\hamil(\coParam')$ and $\hamil_{\overline{SS}}(\coParam')$, we investigate the resultant
\begin{equation} \label{eq:reducedResultant}
R\Big(\frac{det(\hamil(\coParam') - \mathbf{I} \lambda)}{(\lambda - E_{b})^2}, \frac{det(\hamil_{\overline{SS}}(\coParam') - \mathbf{I} \lambda)}{(\lambda - E_{b})^2} \Big) = \delta^4.
\end{equation}
Since $\delta = a^2 + b^2 >0$, $\hamil(\coParam'),\hamil_{\overline{SS}}(\coParam')$ do not share any other roots.

To extract $P_{\pm}(\coParam',\lambda)$, we test whether $p_{\pm}(\coParam',\lambda)$, given by the respective numerators of \cref{eq:example1Pminus,eq:example1Pplus}, have leading order coefficients $+1$ and that $p_{\pm}(\coParam',\lambda)/q_{\pm}(\coParam',\lambda)$ are irreducible fractions.
The latter is indeed the case, but the leading order coefficients are $-1$.
Thus, we have $P_{\pm}(\coParam',\lambda) = -p_{\pm}(\coParam',\lambda)$, and explicitly
\begin{align*}
	P_{+}(\coParam',\lambda) &= E_b \left(h-\lambda +E_r\right)-\lambda  (h-\lambda )-\delta -\lambda  E_r, \\
	P_{-}(\coParam',\lambda) &= -E_b \left(h+\lambda -E_r\right)+\lambda  (h+\lambda )-\delta -\lambda E_r .
\end{align*}

\subsubsection{Enforcing pretty good state transfer}
Inserting $P_{\pm}(\coParam',\lambda)$ into \cref{eq:PGSTTraceCondition} and simplifying the resulting inequality yields
\begin{equation}
	\frac{2}{E_b-h+E_r}\neq \frac{2}{E_b+h+E_r}
\end{equation}
which is obviously fulfilled whenever $h\ne 0$. The only task left is to search for realizations $\coParam'' \in \coParam$ which render both $P_{\pm}(\coParam'',\lambda)$ to be irreducible over the base field $F$ which contains all entries of $\hamil(\coParam'')$.
If we choose $(\coParam'')$ such that $\hamil(\coParam'') \in \mathbb{Q}^{6 \times 6}$, we have $F = \mathbb{Q}$, and one realization leading to PGST is
\begin{equation*}
	a=1,b=2,c=1/4, h=1,E_{b} = E_{r} = 0 .
\end{equation*}

\section{Storage and pretty good transfer of compact localized states} \label{sec:CLSTransfer}
\begin{figure} 
	\centering
	\includegraphics[max size={0.6\columnwidth}{\textheight}]{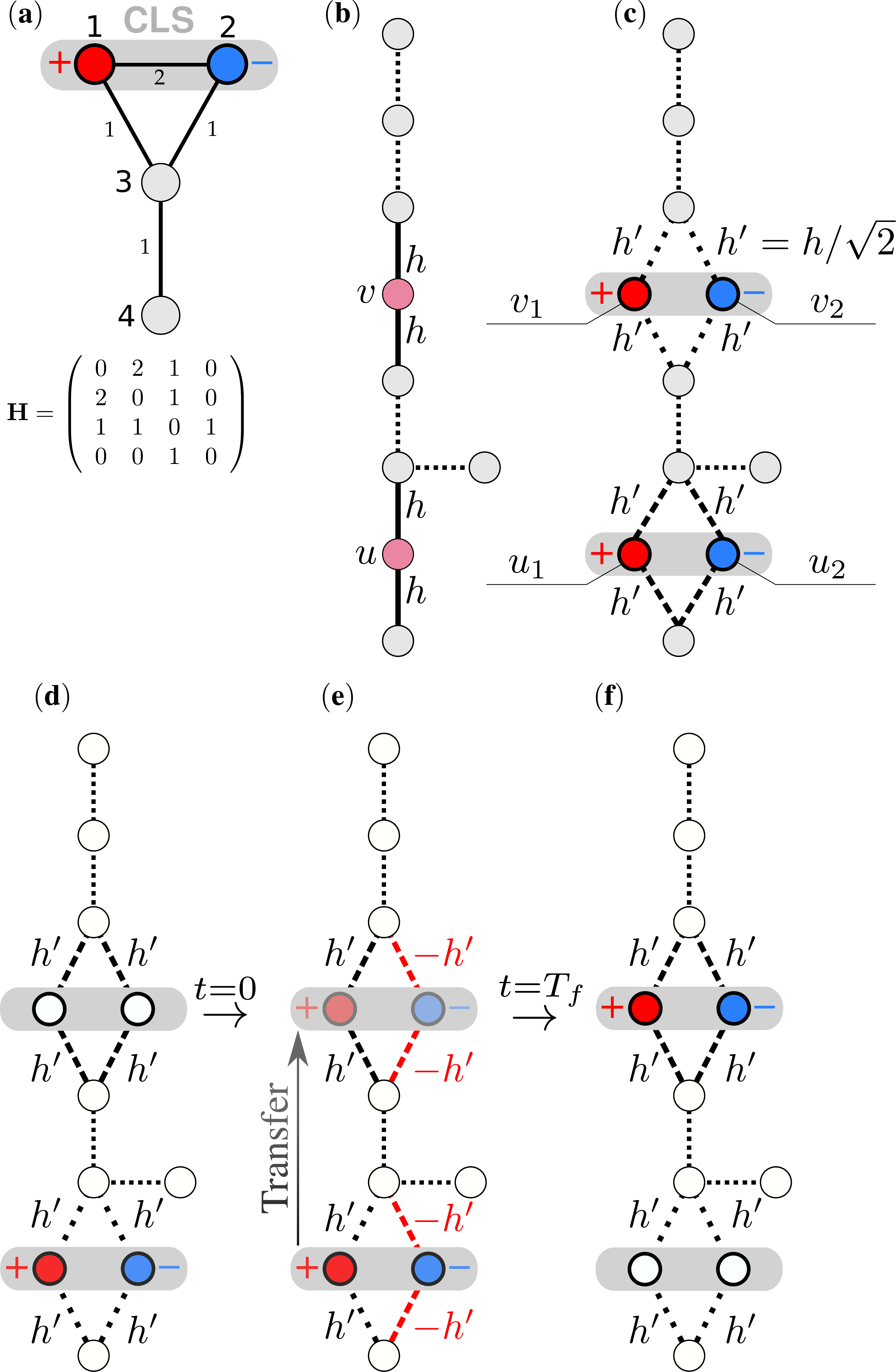}
	\caption{The graph shown in \textbf{(a)} features a compact localizes state (CLS) on the upper two sites. \textbf{(b -- c)} show how a setup featuring PGST of single site excitations $\ket{u}$ and $\ket{v}$ (and without direct coupling between $u$ and $v$) can be equipped with compact localized states by dimerizing sites $u$ and $v$. The setup in \textbf{(b)} is described by $\hamil$, while the setup in \textbf{(c)} is described by $\hamil_{\text{m}}$ (see text for details).
	For each dimer $u_{1,2}$ and $v_{1,2}$, this setup then features one CLS.
	\textbf{(d -- f)} visualize the proposed method of transferring the CLS by performing two quenches at $t=0$ and $t=T_{f}$ (see text for details), so that the three setups are described by $\hamil_{\text{m}}, \hamil'_{\text{m}}$ and again $\hamil_{\text{m}}$, respectively.
	}
	\label{fig:CLSDemonstration}
\end{figure}
\begin{figure} 
	\centering
	\includegraphics[max size={0.55\columnwidth}{\textheight}]{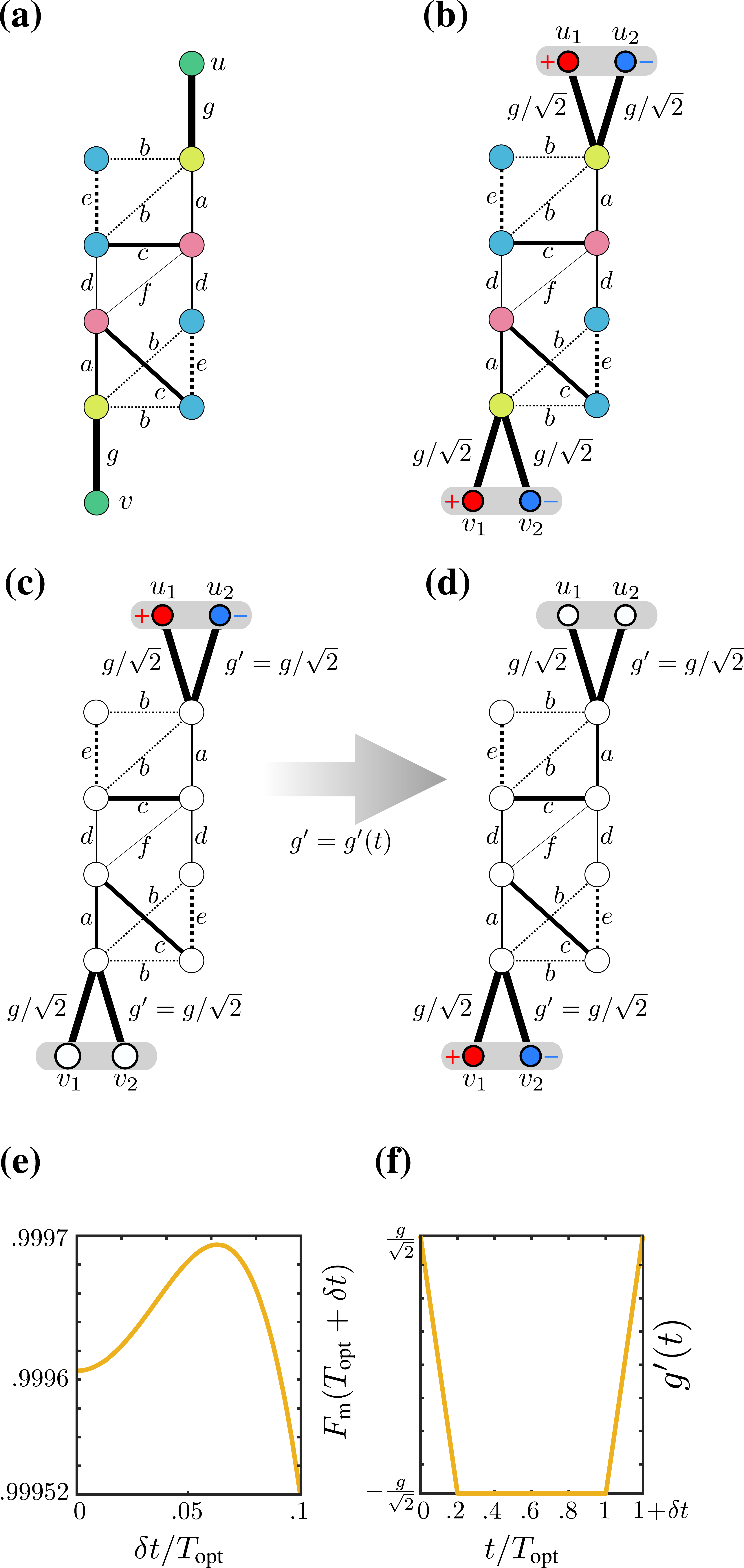}
	\caption{\textbf{(a)} Original setup featuring PGST from $u$ to $v$, described by $\hamil$. \textbf{(b)} Modified setup described by $\hamil_{\text{m}}$, featuring two CLS $\ket{u_{-}} = \frac{\ket{u_{1}} - \ket{u_{2}}}{\sqrt{2}}$ and $\ket{v_{-}} = \frac{\ket{v_{1}} - \ket{v_{2}}}{\sqrt{2}}$, which can be pretty well transferred into each other by performing instantaneous coupling flips (not shown). \textbf{(c -- d)} demonstrates the transfer from $\ket{u_{-}}$ to $\ket{v_{-}}$ by linear ramps of $g'(t)$ with ramping time $\delta t$ [shown in \textbf{(f)} for $\delta t = t_{\text{Opt}}/5$]. The fidelity of this process is shown in \textbf{(e)}.
	}
	\label{fig:CLSTransfer}
\end{figure}

So far, we investigated the transfer of single-site excitations, and showed how networks supporting pretty good transfer of these states can be designed.
In the following, we will demonstrate how such networks can be modified to allow for robust \emph{storage} of qubits.
The need for such modifications arises since, although relatively easy to transfer, single site excitations are difficult to store.
To achieve storage, the underlying sites would need to be completely decoupled from the remainder of the Hamiltonian right after state transfer, which is usually not achievable.
As a consequence, the single site excitation would tunnel to adjacent (weakly) coupled sites, drastically degrading the storage performance.
Recently, a solution to this problem has been proposed in Ref. \cite{Rontgen2018aQuantumNetworkTransferStorage}.
There, qubits were not encoded into excitations of single sites, but into excitations of dimers, which are schematically shown in \cref{fig:CLSDemonstration} \textbf{(a)}, where the dimer consisting of the upper two sites $1$ and $2$ is excited.
Storage in these dimers does not rely on decoupling, but rather on destructive interference.
To achieve such interference, the couplings of the two constituents of the dimer to the remainder of the system are chosen symmetrical, and the dimer-sites are excited with a phase difference of $\pi$.
This completely suppresses any tunneling of this dimer state to its environment. It can be easily proven that such dimer states are eigenstates of the underlying Hamiltonian, and due to their strictly limited spatial extent, they are known as \emph{compact localized states} (CLS).
In \cref{fig:CLSDemonstration} \textbf{(a)}, for example, the CLS is given by $\ket{\Psi_{\text{CLS}}} = \frac{\ket{1} - \ket{2}}{\sqrt{2}}$, and one can easily show that it is an eigenstate of $\hamil$ with eigenvalue $\lambda = -2$. Importantly, this is also an eigenvalue of the \emph{isolated} Hamiltonian of the subsystem
\begin{equation*}
	\hamil_{SS} = \begin{pmatrix}
	0 & 2\\
	2 & 0
	\end{pmatrix}
\end{equation*}
where $S = \{1,2\}$.
The fact that the eigenvalues of CLSs depend only on the subsystem on which they are localized is indeed a general property, and this is just one of the many intriguing features of these states.
Not only do CLSs feature localization without disorder, as is the case for the well-known Anderson localization \cite{Anderson1958PR1091492AbsenceDiffusionCertainRandom}, but they are also strongly connected to the appearance of flat bands.
These are, in turn, conjectured to play a role in the superconduction of cuprates \cite{Leykam2018AP31473052ArtificialFlatBandSystems,Kopnin2011PRB83220503HightemperatureSurfaceSuperconductivityTopological,Iglovikov2014PRB90094506SuperconductingTransitionsFlatbandSystems,Peotta2015NC68944SuperfluidityTopologicallyNontrivialFlat,Julku2016PRL117045303GeometricOriginSuperfluidityLiebLattice,Kobayashi2016PRB94214501SuperconductivityRepulsivelyInteractingFermions,Tovmasyan2016PRB94245149EffectiveTheoryEmergentSU2,Liang2017PRB95024515BandGeometryBerryCurvature}.
We refer the reader interested in the exciting field of CLSs and flat bands to the review \cite{Leykam2018AP31473052ArtificialFlatBandSystems}.

What makes compact localized states important in the context of this work is their unique combination of favorable properties.
The fact that they are eigenstates allows for their perfect, i.e., unity fidelity, storage in idealized model systems, where imperfections can be ignored.
If, on the other hand, such model systems are realized and imperfections are introduced, CLSs profit from the fact that they are localized only on a subdomain of the full system.
This means that they are \emph{immune} to any imperfections of the underlying Hamiltonian outside of this domain and its directly neighboring sites. Morever, the fact that they are localized by means of destructive interference means that they are even immune to certain perturbations inside or directly next to their domain of localization.
For example, the coupling $\hamil_{1,2}$ in \cref{fig:CLSDemonstration} \textbf{(a)} could be varying in time, but would only give an overall time-varying phase on $\ket{\Psi_{\text{CLS}}}$, which would still remain a compactly localized eigenstate of $\hamil(t)$.
Moreover, the couplings $\hamil_{1,3}$ and $\hamil_{2,3}$ could be chosen \emph{arbitrarily big and also time-dependent}, but $\ket{\Psi_{\text{CLS}}}$ would be completely unaffected as long as $\hamil_{1,3}(t) = \hamil_{2,3}(t)$ for all $t$.

The combination of all these properties clearly renders compact localized states ideal candidates for the storage of qubits.
However, the fact that they are eigenstates of $\hamil$ complicates their transfer, which is naturally impossible by simple time-evolution if $\hamil$ is time-independent.
In Ref. \cite{Rontgen2018aQuantumNetworkTransferStorage}, a set of minimal changes to the setup have been demonstrated that allow for both perfect storage and perfect, i.e., unity fidelity, transfer of CLSs in specialized networks.
In this section, we use the underlying idea and show how a network capable of PGST of single-site excitations of sites $u$ and $v$ can be modified by a set of minimal changes such that (i) the network supports compact localized states and (ii) it is possible to perform pretty good transfer of these states.
The only condition on the underlying network is that there are no direct links (edges) between $u$ and $v$.
The basic idea is sketched in \cref{fig:CLSDemonstration} \textbf{(b -- f)}. We start from a Hamiltonian $\hamil$ [as the one depicted in \cref{fig:CLSDemonstration} \textbf{(b)}] which supports PGST from $u$ to $v$, with time-dependent fidelity
\begin{equation}
	F(t) = |\braket{u|exp(i \hamil t)|v}|^2.
\end{equation}
We then modify $\hamil$ such that $u$ and $v$ are replaced by dimers $u_{1,2}$ and $v_{1,2}$, and all couplings of $u,v$ to their environment are replaced by symmetrized and renormalized couplings with the dimer, as shown in \cref{fig:CLSDemonstration} \textbf{(c)}.

The fidelity
\begin{equation} \label{eq:symmetrizedFidelity}
	F'(t) = |\braket{u_{+}|exp(i \hamil_{\text{m}} t)|v_{+}}|^2
\end{equation}
for the transfer of symmetric excitations $\ket{u_{+}} = \frac{\ket{u_{1}} + \ket{u_{2}}}{\sqrt{2}}$ to $\ket{v_{+}} = \frac{\ket{v_{1}} + \ket{v_{2}}}{\sqrt{2}}$ by means of the modified Hamiltonian $\hamil_{\text{m}}$ can then be shown [see \cref{app:CLSTransfer} for details] to be \emph{identical} to $F(t)$.
In particular, while $\hamil$ supports PGST of single site excitations $u$ and $v$, its modified version $\hamil_{\text{m}}$ supports PGST of \emph{symmetric} dimer excitations $\ket{u_{+}}$ and $\ket{v_{+}}$.

On the other hand, $\hamil_{\text{m}}$ supports also two compact localized states, $\ket{u_{-}} = \frac{\ket{u_{1}} - \ket{u_{2}}}{\sqrt{2}}$, and $\ket{v_{-}} = \frac{\ket{v_{1}} - \ket{v_{2}}}{\sqrt{2}}$. They are eigenstates of $\hamil_{\text{m}}$ and can thus not be transferred by simple time evolution. However, by suitable time-dependent modifications, $\ket{u_{-}}$ can be pretty well transferred to $\ket{v_{-}}$, as we show in the following. The procedure is visualized in \cref{fig:CLSDemonstration} \textbf{(d -- f)}.
The main idea is to achieve such a transfer by performing two quenches
\begin{equation*}
	\hamil_{\text{m}} \overset{t=0}{\rightarrow} \hamil'_{\text{m}} \overset{t=T_{f}}{\rightarrow} \hamil_{\text{m}}
\end{equation*}
at $t=0$ and $t=T_{f}$.
The Hamiltonian $\hamil'_{\text{m}}$ is constructed from $\hamil_{\text{m}}$ by instantaneously switching all couplings $J_{i} = h_{i,u_{1}} = h_{i,v_{1}}$ of $u_{1},v_{1}$ (but not of $u_{2},v_{2}$) to their environment as $J_{i} \rightarrow - h_{i}$. Due to this change, the CLS $\ket{u_{-}}$ is no longer an eigenstate of $\hamil'_{\text{m}}$, and thus spreads across the lattice. 
The transfer fidelity during this spreading is given by
\begin{equation}
	F''(t) = |\braket{u_{-}|exp(i \hamil'_{\text{m}} t)|v_{-}}|^2 = F'(t) = F(t) .
\end{equation}
Once $F''(t)$ achieves a sufficiently high value $F''(T_{f}) = 1 - \epsilon$ for given $\epsilon$, the second quench is performed, and the previously modified couplings are instantaneously switched back to their original value.
At $T_{f}$, the state of the system is then given by
\begin{equation}
	\ket{\Psi(T_{f})} = e^{i \phi} \sqrt{1 - \epsilon} \ket{v_{-}} + \sum_{\nu}^{} c_{\nu} \ket{\psi^{\nu}}
\end{equation}
where $\phi$ is a phase and the coefficients $c_{\nu}$ must fulfill $|\braket{\Psi(T_{f})|\Psi(T_{f})}| = 1$.
The states $\ket{v_{-}}$ and $\ket{\psi_{\nu}} \ne \ket{v_{-}}$ are eigenstates of the pre/post quench Hamiltonian $\hamil_{\text{m}}$. Since $\hamil(t) = \hamil_{\text{m}}$ for $t\ge T_{f}$, we have
\begin{equation}
	|\braket{v_{-} | \Psi(t\ge T_{f})}|^2 = F''(T_{f}) = \text{constant}.
\end{equation}
The CLS $\ket{u_{-}}$ is thus stored with the time-independent fidelity $F''(t_{f})$ and, due to its properties, enjoys protection against a large number of imperfections of $\hamil_{\text{m}}$.

In practice, instantaneous coupling flips are rather unrealistic, and may be replaced by more realistic switching pulses.
These will naturally change the transfer fidelity, and the strength of this change clearly depends both on the individual system and the realization of the flipping pulse.
In Ref. \cite{Rontgen2018aQuantumNetworkTransferStorage}, the impact of linear ramps (instead of instantaneous coupling flips) on linear chains that support perfect transfer of compact localized states has been investigated.
As has been shown there for the case of chains of length $N = 5$, even extraordinary slow ramping times of nearly half of the total transfer time only leads to a decrease of the transfer fidelity from unity to $0.97$.
This being said, we now exemplarily investigate the impact of finite duration linear ramps of couplings on the transfer fidelity of the simple example setup shown in \cref{fig:CLSTransfer} \textbf{(a)}. The two green sites $u$ and $v$ are cospectral for any choice of the $11$ parameters $\coParam = \{a,b,c,d,e,E_{red},E_{blue},E_{green},E_{yellow}\} \in \mathbb{R}$.
Before investigating the impact of finite-time ramps on the transfer of compact localized states, we first find the subspace $\coParam'' \subseteq \coParam$ in which the setup supports PGST of single site excitations.
Within this subspace, we then look for realizations $\coParam''' \subseteq \coParam''$ for which the maximum transfer fidelity
\begin{equation*}
	F_{\text{max}}(T_{f}) = \text{max}(F(t\le T_{f}))
\end{equation*}
from site $u$ to $v$
within a given time $T_{f}$ and boundaries on the absolute values of parameters $\coParam''$ is as large as possible.
In other words, we optimize the system to (i) support PGST from $u$ to $v$ and (ii) reach an acceptable transfer fidelity in as little time as possible.
In practical applications, such an optimization is always necessary.
Since PGST by definition is an asymptotic property, the underlying network may reach a suitably high transfer fidelity only after prohibitively long transfer times.
For the setup shown in \cref{fig:CLSTransfer} \textbf{(a)}, we restricted the optimization to the subspace where all on-site potentials vanish, and obtained a maximum transfer fidelity
\begin{equation*}
	F(T_{\text{opt}} = 10.8345) = 0.996
\end{equation*}
for $a=0.7975,\; b=0.8103, \; c=0.8880,\; d=2.3473,\; e=2.3005,\; f=0.3061,\; g=0.5489$.

In order to investigate the transfer of compact localized states, we first apply the above set of modifications to \cref{fig:CLSTransfer} \textbf{(a)} and equip it with two compact localized states. The modified setup is shown in \cref{fig:CLSTransfer} \textbf{(b)} and supports the two CLS $\ket{u_{-}} = \frac{\ket{u_{1}} - \ket{u_{2}}}{\sqrt{2}}$ and $\ket{v_{-}} = \frac{\ket{v_{1}} - \ket{v_{2}}}{\sqrt{2}}$.
By performing instantaneous coupling flips at $t=0$ and $T_{\text{opt}}$, we can transfer $\ket{u_{-}}$ to $\ket{v_{-}}$ (and vice versa)
with the fidelity $F(T_{\text{opt}})$. We now slightly change the protocol and switch the couplings by performing linear \emph{ramps} with a duration $\delta t$. The ramps are started at $t=0$ and $t=T_{\text{opt}}$, so that the transfer process is finished at $t=T_{\text{opt}} + \delta t$. The process is sketched in \cref{fig:CLSTransfer} \textbf{(c)} and \textbf{(d)}. \Cref{fig:CLSTransfer} \textbf{(c)} shows the setup at $t=0$, where the state of the system is given by $\ket{\Psi(t=0)} = \ket{u_{-}}$ (white circles denoting sites with zero amplitude). The transfer process is then started by linearly ramping down $g'(t)$  such that $g'(t=\delta t) = -g/\sqrt{2}$. At $t=T_{\text{opt}}$, these are then linearly ramped up again, reaching their final value $g'(T_{\text{opt}} + \delta t) = g/\sqrt{2}$. The pulse $g'(t)$ is shown in \cref{fig:CLSTransfer} \textbf{(f)} for $\delta t = T_{\text{opt}}/10$. \Cref{fig:CLSTransfer} \textbf{(d)} shows the final situation, where white sites now denote very low (but not necessarily zero) amplitudes of the final state $\ket{\Psi(T_{\text{opt}} + \delta t)}$.
In \cref{fig:CLSTransfer} \textbf{(e)}, the transfer fidelity is plotted against the pulse duration $\delta t$.
Quite counter-intuitively, the fidelity of transferring compact localized states \emph{increases} first for increasing $\delta t$.
Investigating the cause for this behavior would certainly be a worthwhile topic for further research.
For larger $\delta t$, the fidelity falls off as expected, but overall remains quite high. Even for comparatively slow ramps of $\delta t = T_{\text{opt}}/10$, the transfer fidelity decreases only by roughly $10^{-4}$. Notably, this high robustness against slow control pulses was also observed in Ref. \cite{Rontgen2018aQuantumNetworkTransferStorage} for the case of linear chains equipped with compact localized states.

\section{Brief Conclusion} \label{sec:conclusions}

We presented a method to design Hamiltonians $\hamil$ featuring pretty good state transfer (PGST) between two sites.
A necessary condition for PGST is that these two sites are so-called strongly cospectral, which means that all eigenstates have parity $\pm 1$ on these two sites.
We showed how Hamiltonians featuring strongly cospectral sites can be designed.
We then relied on so-called isospectral reductions of these Hamiltonians to yield a factoring of their characteristic polynomial in terms of smaller polynomials $P_{\pm}$, which are related to eigenvectors with parity $\pm 1$ on $u$ and $v$.
The motivation for this factorization is the fact that PGST automatically arises in setups where the coefficients of $P_{\pm}$ fulfill a set of relations, as has recently been shown by Eisenberg et al. \cite{Eisenberg2018DMPrettyGoodQuantumState}.
Equipped with explicit knowledge of $P_{\pm}$, we show how they can be properly manipulated by changing couplings and on-site potentials whilst maintaining the strong cospectrality.
Through these manipulations, PGST can therefore be achieved in certain setups, and we develop our method into an algorithm to design PGST Hamiltonians.
We further show how Hamiltonians featuring PGST can be equipped with so-called compact localized states (CLS).
Such states are eigenstates of $\hamil$ and are strictly localized on a spatially finite (and usually very small) domain, which allows for robust storage of qubits encoded into such CLSs. We further present time-dependent protocols which allow for PGST of CLSs.
Our work opens new routes towards flexible design of PGST networks and broadens their scope to allow for robust storage as well.
An important future task is to investigate how well the transfer fidelity of PGST Hamiltonians within a given maximal transfer time $T_{\text{max}}$ can be optimized by parameter tuning.
This task should be supported by the algorithm presented in this work, as it allows to obtain rather small parameter spaces $\coParam$ in which a given parameter dependent Hamiltonian $H(\coParam)$ features PGST.

\section{Acknowledgments}
	M.R. gratefully acknowledges financial support by the `Stiftung der deutschen Wirtschaft' in the framework of a scholarship.
	N.E.P. gratefully acknowledges financial support from the Hellenic Foundation for Research and Innovation (HFRI) and the General Secretariat for Research and Technology (GSRT) under the HFRI PhD Fellowship Grant No. 868.
	I.B. acknowledges financial support by Greece and the European Union
	(European Social Fund -- ESF) through the Operational Programme ``Human Resources Development, Education and Lifelong Learning'' in the context of the project ``Reinforcement of Postdoctoral Researchers'' (MIS-5001552), implemented by the State Scholarships Foundation (IKY). M.P. gratefully acknowledges financial support by the `Studienstiftung des deutschen Volkes' in the framework of a scholarship.

\appendix

\section{Proof for the interpretation of matrix entries of powers of $\hamil^{k}$ in terms of walks} \label{app:walkInterpretationProof}
We now prove \cref{eq:walkInterpretation} which states that
\begin{equation} \label{eq:walkInterpretationRepetition}
(\hamil^{k>0})_{a,b} = \sum_{p} w\left( p_{a,b}^{(k)} \right)
\end{equation}
where $w\left( p_{a,b}^{(k)} \right)$ denotes the weight of one possible walk of length $k$ between vertices $a$ and $b$, and the sum is over all such walks.
To this end, we write $(\hamil^{k>0})_{a,b}$ as
\begin{equation} \label{eq:walkInterpretationExpansion}
	(\hamil^{k>0})_{a,b} = \sum_{l_{1},\ldots{},l_{k-1}} \hamil_{a,l_{1}} \hamil_{l_{1},l_{2}} \ldots \hamil_{l_{k-2},l_{k-1}} \hamil_{l_{k-1},b}
\end{equation}
where each index $l_{i}$ goes from $1$ to $N$ with $\hamil \in \mathbb{R}^{N \times N}$.
We now interpret every term $\hamil_{i,j}$ occurring in \cref{eq:walkInterpretationExpansion} as the weight of the edge connecting sites $i$ and $j$.
Each summand is, therefore, the weight of a walk of length $k$ from site $a$ to $b$ via the sites $l_{1},l_{2},\ldots{},l_{k-1}$, where walks over physically non-existing edges (i.e., those with vanishing weights $\hamil_{i,j} = 0$) naturally have vanishing weights as well.
As a consequence, we can write \cref{eq:walkInterpretationExpansion} as \cref{eq:walkInterpretationRepetition}, and the value of the matrix element $(\hamil^{k})_{a,b}$ is equal to the sum of weights of all walks of length $k$ between vertices $a$ and $b$.

\section{Proofs for cospectrality} \label{app:CospectralityProofs}
\begin{figure} 
	\centering
	\includegraphics[max size={.5\columnwidth}{\textheight}]{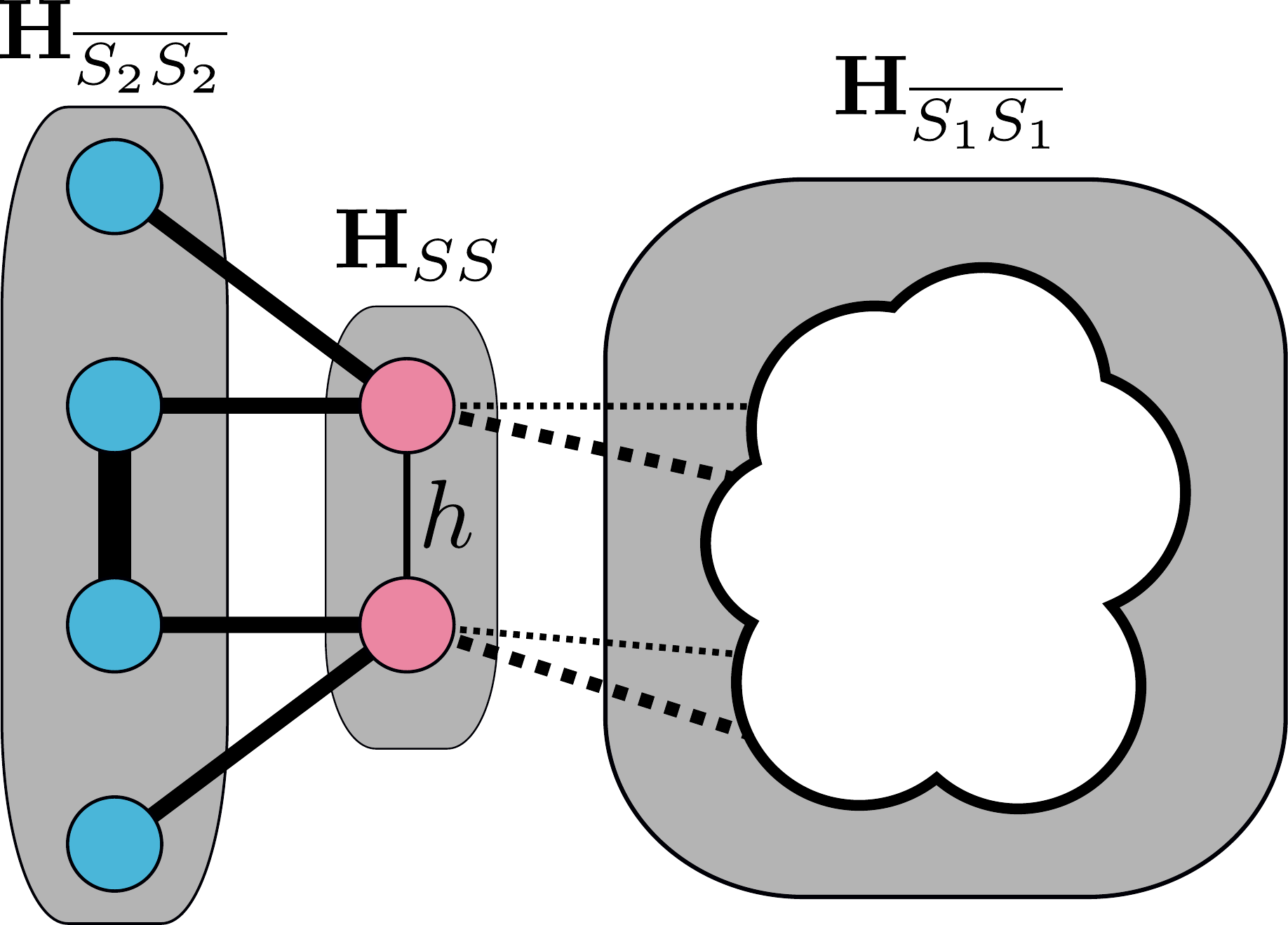}
	\caption{Visualization of the decomposition of $\hamil$ into subsystems.
	}
	\label{fig:proofCospectrality}
\end{figure}
We now prove the validity of the design mechanism presented in \cref{sec:designingCospectral}. In particular, we will prove that all changes shown in \cref{fig:cospectralityDesign} applied onto an already cospectral network keep this cospectrality.
To this end, we proof the following
\begin{theorem}
	Let
	\begin{equation}
	\hamil_{i} = \begin{pmatrix}
	\hamil_{SS} & \hamil_{S\overline{S_{i}}} \\
	\hamil_{\overline{S_{i}}S} & \hamil_{\overline{S_{i}S_{i}}}
	\end{pmatrix} \in \mathbb{R}
	\end{equation}
	be symmetric matrices with $i=1,2$ and
	\begin{equation*}
	\hamil_{SS} = \begin{pmatrix}
	E & h \\
	h & E
	\end{pmatrix}
	\end{equation*}
	bisymmetric. Denote the two sites in $\hamil_{SS}$ as $S = \{u,v\}$.
	If $u,v$ are cospectral in $\hamil_{i}$, then they are also cospectral in
	\begin{equation*} 
		\hamil' = \begin{pmatrix}
		\hamil_{SS} & \hamil_{S\overline{S_{1}}} & \hamil_{S\overline{S_{2}}} \\
		\hamil_{\overline{S_{1}}S} & \hamil_{\overline{S_{1}S_{1}}} & \mathbf{0} \\
		\hamil_{\overline{S_{2}}S} & \mathbf{0}& \hamil_{\overline{S_{2}S_{2}}}
		\end{pmatrix} .
	\end{equation*}
\end{theorem}
\begin{proof}
	We use the fact \cite{Kempton2019AMCharacterizingCospectralVerticesIsospectral} that the isospectral reduction $\mathcal{\mathbf{R}}_{S=\{u,v\}}(\hamil,\lambda)$ is, for symmetric $\hamil$, bisymmetric if and only if the sites $u$ and $v$ are cospectral. Individually, we therefore have that
	\begin{equation*}
		\mathcal{\mathbf{R}}_{S}(\hamil_{i},\lambda) = \hamil_{SS} - \hamil_{S\overline{S_{i}}}\left(\hamil_{\overline{S_{i}S_{i}}} - \lambda \mathbf{I} \right)^{-1} \hamil_{\overline{S_{i}}S}
	\end{equation*}
	is bisymmetric.
	We then evaluate the isospectral reduction of $\hamil'$, which can be written as
\begin{equation} \label{eq:decompositionOfIsospectral}
\mathcal{\mathbf{R}}_{S}(\hamil',\lambda) = \hamil_{SS} - \mathbf{A} \mathbf{B}^{-1} \mathbf{A}^{T}
\end{equation}
where
\begin{equation*}
	\mathbf{A} = (\hamil_{S\overline{S_{1}}}\,,\; \hamil_{S\overline{S_{2}}}),\;\; \mathbf{B} = \begin{pmatrix}
	\hamil_{\overline{S_{1}S_{1}}} - \lambda \mathbf{I} & \mathbf{0} \\
	\mathbf{0} & \hamil_{\overline{S_{2}S_{2}}} - \lambda \mathbf{I}
	\end{pmatrix} .
\end{equation*}
\Cref{eq:decompositionOfIsospectral} then becomes
\begin{equation*}
	\mathcal{\mathbf{R}}_{S}(\hamil',\lambda) = \hamil_{SS} - \sum_{i} \hamil_{S\overline{S_{i}}}\left(\hamil_{\overline{S_{i}S_{i}}} - \lambda \mathbf{I} \right)^{-1} \hamil_{\overline{S_{i}}S} .
\end{equation*}
This expression is bisymmetric, since $\hamil_{SS}$ as well as each of the two summands are individually bisymmetric, and sums of bisymmetric matrices are bisymmetric again. Due to the connection between bisymmetry of $\mathcal{\mathbf{R}}_{S = \{u,v\}}(\hamil,\lambda)$ and the cospectrality of sites $u$ and $v$, we have therefore proven the above theorem.
\end{proof}
To apply this theorem to \cref{sec:designingCospectral}, we divide the setups shown their into three parts, as shown in \cref{fig:proofCospectrality} for a slightly modified version of the graph depicted in \cref{fig:cospectralityDesign} \textbf{(b2)}. Given the cospectrality of the two red sites, denoted by $u$ and $v$, in the original system $\textbf{H}_{1}$, the setup is modified by adding the subsystem containing the sites $\overline{S_{2}}$.
By the above theorem, we know that the cospectrality of $u$ and $v$ is kept also in the composite system $\textbf{H}'$ provided that $\mathcal{\mathbf{R}}_{S}(\textbf{H}_{2},\lambda)$ is bisymmetric.
By explicitly computing the isospectral reductions for each graph presented in \cref{fig:cospectralityDesign}, it can be proven that all of them are cospectral. We have thus proven the validity of the design mechanism proposed in \cref{sec:designingCospectral}.

\section{Proof for the form of $P_\mathbf{R}^{\pm}$} \label{app:ProofOfPPrime}
We want to prove that, for bisymmetric
\begin{equation}
\mathcal{\mathbf{R}}_{S}(\hamil,\lambda) = \begin{pmatrix}
A(\lambda) & B(\lambda) \\
B(\lambda) & A(\lambda)
\end{pmatrix}
\end{equation}
the characteristic polynomials $P_{\mathbf{R}}^{\pm}(\lambda)$ related to eigenvectors of positive and negative parity, respectively, are given by $P_{\mathbf{R}}^{\pm}(\lambda) = A(\lambda) \pm B(\lambda) - \lambda$.

To prove this, we perform a similarity transform $\mathbf{R}'_{S}(\hamil,\lambda) = A^{-1} \mathbf{R}_{S}(\hamil,\lambda) A$, with
\begin{equation*}
A = \begin{pmatrix}
1 & 1\\
1 & -1
\end{pmatrix}
\end{equation*}
so that
\begin{equation*}
\mathbf{R}'_{S}(\hamil,\lambda) = \begin{pmatrix}
A(\lambda) + B(\lambda) & 0\\
0 & A(\lambda) - B(\lambda)
\end{pmatrix}
\end{equation*}
becomes block-diagonal. Therefore, its eigenvectors are obviously $(1,0)^{T}$ [those of the first block] with eigenvalues $\{\lambda^{1}_{i}\}$ and $(0,1)^{T}$ [those of the second block] with eigenvalues $\{\lambda^{2}_{j}\}$. Multiplying these eigenvectors by $A$ then yields the corresponding eigenvectors of $\mathcal{\mathbf{R}}_{S}(M,\lambda)$.
Therefore, these are obviously $(1,1)^{T}$ with eigenvalues $\{\lambda^{1}_{i}\}$ and $(1,-1)^{T}$ with eigenvalues $\{\lambda^{2}_{j}\}$.
We remind the reader that, since $\mathcal{\mathbf{R}}_{S}(M,\lambda)$ depends on $\lambda$, it can have \emph{more} then two eigenvectors, and these need not be linearly independent.
The eigenvalues $\{\lambda^{1}_{i}\}, \{\lambda^{2}_{j}\}$ are therefore related to eigenvectors of positive and negative parity, respectively, and are the solutions to the equations
\begin{align*}
det\Big( A(\lambda^{1}_{i}) + B(\lambda^{1}_{i}) - \lambda^{1}_{i} \Big) &= 0\\
det\Big( A(\lambda^{2}_{j}) - B(\lambda^{2}_{j}) - \lambda^{2}_{j} \Big) &= 0 .
\end{align*}
It is thus obvious that $P_{\mathbf{R}}^{\pm}(\lambda) = A(\lambda) \pm B(\lambda) - \lambda$ are the characteristic polynomials related to eigenvectors of positive and negative parity, respectively, and that
\begin{equation*}
	P_{\mathbf{R}}^{+} \cdot P_{\mathbf{R}}^{-} = det\big(\mathcal{\mathbf{R}}'_{S}(H) \big) = det\big(\mathbf{R}_{S}(H) \big) .
\end{equation*}

\section{Mathematical details on the transfer of compact localized states} \label{app:CLSTransfer}
We now prove the statements made in \cref{sec:CLSTransfer}. The proofs are similar to those done in \cite{Rontgen2018aQuantumNetworkTransferStorage}, but are included here so that the current work is self-contained.

We assume that the original network is described by a Hamiltonian $\hamil$ and supports PGST between sites $S = \{u,v\}$.
We then partition the system such that
\begin{equation*}
	\hamil = \begin{pmatrix}
	\hamil_{\overline{SS}}  & \hamil_{S\overline{S}} \\
	\hamil_{\overline{S}S} & \hamil_{SS}
	\end{pmatrix} \in \mathbb{R}^{(N+2) \times (N+2)}.
\end{equation*}
As stated in \cref{sec:CLSTransfer}, we demand $\hamil$ to have no direct coupling between $u$ and $v$, so that
\begin{equation*}
	\hamil_{SS} = \begin{pmatrix}
	E & 0 \\
	0 & E
	\end{pmatrix}  \in \mathbb{R}^{2 \times 2}
\end{equation*}
is diagonal.
We denote the eigenvectors of $\hamil$ as
\begin{equation*}
\ket{\phi^{\nu}} = \begin{pmatrix}
\mathbf{w}^{\nu}\\
x_{u}^{\nu}\\
x_{v}^{\nu}
\end{pmatrix}\;,\; \in \mathbb{R}^{(N+2) \times 1}
\end{equation*}
with $\mathbf{w}^{\nu} \in \mathbb{R}^{N\times{}1}$.
The fidelity for transfer from $\ket{u}$ to $\ket{v}$ is given as
\begin{equation*}
|\braket{u|e^{i \hamil t}|v}|^2 = \left| \sum_{\nu}^{} x_{u}^{\nu} (x_{v}^{\nu})^{*} e^{i \lambda_{\nu} t}\right|^2 .
\end{equation*}

We then modify the system as shown in \cref{fig:CLSDemonstration} \textbf{(c)}, so that its Hamiltonian becomes
\begin{equation*}
\hamil_{\text{m}} = \begin{pmatrix}
\hamil_{\overline{SS}} & \frac{1}{\sqrt{2}} \hamil_{S\overline{S}} & \frac{1}{\sqrt{2}} \hamil_{S\overline{S}}\\
\frac{1}{\sqrt{2}} \hamil_{\overline{S}S} & \hamil_{SS} & \mathbf{0}_{2\times{}2}\\
\frac{1}{\sqrt{2}} \hamil_{\overline{S}S} & \mathbf{0}_{2\times{}2} & \hamil_{SS}
\end{pmatrix}.
\end{equation*}
By means of the \emph{`equitable partition theorem'} \cite{Barrett2017LAIA513409EquitableDecompositionsGraphsSymmetries,Francis2017LAIA532432ExtensionsApplicationsEquitableDecompositions,Rontgen2018PRB97035161CompactLocalizedStatesFlat} its $N+4$ eigenstates can then be shown to be 
\begin{equation*}
\ket{\phi^{\nu}} = \begin{pmatrix}
\mathbf{w}^{\nu}\\
\frac{1}{\sqrt{2}} x_{u}^{\nu}\\
\frac{1}{\sqrt{2}} x_{v}^{\nu}\\
\frac{1}{\sqrt{2}} x_{u}^{\nu}\\
\frac{1}{\sqrt{2}} x_{v}^{\nu}\\
\end{pmatrix},
\ket{\phi^{N+2+r}} = 
\begin{pmatrix}
\mathbf{0}_{N\times{}1}\\
\mathbf{z}^{r}\\
-\mathbf{z}^{r}
\end{pmatrix}
\end{equation*}
with $\nu = 1,\ldots{},N+2$ and $r=1,2$. The $\mathbf{z}^{r}\in \mathbb{C}^{2\times{}1}$ are the eigenvectors of the \emph{isolated} $\hamil_{SS}$.
We now denote the first $N$ sites as $\overline{S}$, and the remaining four as $u_{1},v_{1},u_{2},v_{2}$.
The fidelity
\begin{equation}
F'(t) = |\braket{u_{+}|exp(i \hamil_{\text{m}} t)|v_{+}}|^2
\end{equation}
[\cref{eq:symmetrizedFidelity} from \cref{sec:CLSTransfer}] for the transfer of symmetric excitations $\ket{u_{+}} = \frac{\ket{u_{1}} + \ket{u_{2}}}{\sqrt{2}}$ to $\ket{v_{+}} = \frac{\ket{v_{1}} + \ket{v_{2}}}{\sqrt{2}}$ can then be evaluated as
\begin{align*}
	F'(t) =&\left|\braket{u_{+}|exp(i \hamil_{\text{m}} t)|v_{+}}\right|^2 \\
	=& \left|\sum_{\nu=1}^{N+2} \braket{u_{+}|\phi^{\nu}} \braket{\phi^{\nu}|v_{+}} e^{i \lambda_{\nu} t}\right|^2 \\
	=&\left|\sum_{\nu=1}^{N+2} x_{u}^{\nu} (x_{v}^{\nu})^{*} e^{i \lambda_{\nu} t}\right|^2 \\
	=& F(t) 
\end{align*}
as claimed in \cref{sec:CLSTransfer}, since the overlap of $\ket{u_{+}},\ket{v_{+}}$ with $\ket{\phi^{N + 2 + r}}$ vanishes.

We now look at the compact localized states supported by $\hamil_{\text{m}}$. There are two of these, given by $\ket{I'} = \frac{\ket{u_{1}} - \ket{u_{2}}}{\sqrt{2}}$ (localized on sites $u_{1}$ and $u_{2}$) and $\ket{F'} = \frac{\ket{v_{1}} - \ket{v_{2}}}{\sqrt{2}}$ (localized on sites $v_{1}$ and $v_{2}$).
To transfer $\ket{I'}$ to $\ket{F'}$, we perform an instantaneous flip of couplings at $t=0$, so that
\begin{equation*}
\hamil_{\text{m}} \rightarrow \hamil'_{\text{m}} = \begin{pmatrix}
\hamil_{\overline{SS}} & \frac{1}{\sqrt{2}} \hamil_{S\overline{S}} & -\frac{1}{\sqrt{2}} \hamil_{S\overline{S}}\\
\frac{1}{\sqrt{2}} \hamil_{\overline{S}S} & \hamil_{SS} & \mathbf{0}_{2\times{}2}\\
-\frac{1}{\sqrt{2}} \hamil_{\overline{S}S} & \mathbf{0}_{2\times{}2} & \hamil_{SS}
\end{pmatrix}
\end{equation*}
and $\ket{I'},\ket{F'}$ are no longer eigenstates of $\hamil'_{\text{m}}$. By means of the so-called \emph{`nonequitable partition theorem'} \cite{Fritscher2016SJMAA37260ExploringSymmetriesDecomposeMatrices,Rontgen2018PRB97035161CompactLocalizedStatesFlat} the $N+4$ eigenstates of $\hamil'_{\text{m}}$ can then shown to be 
\begin{equation*}
\ket{\phi^{\nu}} = \begin{pmatrix}
\mathbf{w}^{\nu}\\
\frac{1}{\sqrt{2}} x_{u}^{\nu}\\
\frac{1}{\sqrt{2}} x_{v}^{\nu}\\
-\frac{1}{\sqrt{2}} x_{u}^{\nu}\\
-\frac{1}{\sqrt{2}} x_{v}^{\nu}\\
\end{pmatrix},
\ket{\phi^{N+2+r}} = 
\begin{pmatrix}
\mathbf{0}_{N\times{}1}\\
\mathbf{z}^{r}\\
\mathbf{z}^{r}
\end{pmatrix}
\end{equation*}
with $\nu = 1,\ldots{},N+2$, $r=1,2$, and $\mathbf{z}^{r}$ as above.
We then yield
\begin{align*}
	F''(t) =& |\braket{u_{-}|exp(i \hamil'_{\text{m}} t)|v_{-}}|^2 \\
	=& \left|\sum_{\nu=1}^{N+2} \braket{u_{-}|\phi^{\nu}} \braket{\phi^{\nu}|v_{-}}e^{i \lambda_{\nu} t} \right|^2 \\
	=& \left|\sum_{\nu=1}^{N+2} x_{u}^{\nu} (x_{v}^{\nu})^{*} e^{i \lambda_{\nu} t}\right|^2\\
	=& F'(t) = F(t)
\end{align*}
as claimed in \cref{sec:CLSTransfer}.

Thus, if $\hamil$ supports PGST of single site excitations $u$ and $v$, $\hamil_{\text{m}}$ supports two compact localized states, and by switching $\hamil_{\text{m}} \rightarrow \hamil'_{\text{m}}$, these compact localized states can be pretty well transferred. As explained in \cref{sec:CLSTransfer}, the compact localized states can also be stored with time-independent fidelity $F(t>T_{f}) = F(T_{f})$ by instantaneously switching $\hamil'_{\text{m}} \overset{t=T_{f}}{\rightarrow} \hamil_{\text{m}}$.

\end{document}